\documentclass[aps,prx,reprint,nofootinbib,longbibliography,floatfix,superscriptaddress]{revtex4-2}
\usepackage[normalem]{ulem}
\usepackage[T1]{fontenc}
\usepackage{amsmath,amssymb,amsthm,mathtools}
\usepackage{booktabs}
\usepackage{graphicx}
\graphicspath{{figures/}}
\usepackage{lmodern}
\usepackage{enumitem}
\usepackage[dvipsnames]{xcolor}
\usepackage{tikz}
\usetikzlibrary{arrows.meta,positioning,fit,calc,decorations.pathreplacing,
                backgrounds,shapes.geometric}
\usepackage[colorlinks=true,allcolors=MidnightBlue]{hyperref}

\newtheorem{theorem}{Theorem}
\newtheorem{proposition}[theorem]{Proposition}
\theoremstyle{remark}

\newcommand{\F}{\mathbb F}
\newcommand{\Span}{\operatorname{span}}
\newcommand{\wt}{\operatorname{wt}}
\newcommand{\RM}{\operatorname{RM}}
\newcommand{\rank}{\operatorname{rank}}
\newcommand{\CCZ}{\mathsf{CCZ}}
\newcommand{\CCCZ}{\mathsf{CCCZ}}
\newcommand{\CS}{\mathsf{CS}}
\newcommand{\Tgate}{\mathsf T}
\newcommand{\od}{\mathbin{\odot}}
\newcommand{\one}{\mathbf 1}
\newcommand{\factory}[3]{[[#1,#2,#3]]}
\newcommand{\cK}{\mathcal K}
\newcommand{\cS}{\mathcal S}
\newcommand{\cG}{\mathcal G}
\newcommand{\ket}[1]{\lvert #1\rangle}
\newcommand{\rr}{\raggedright\let\\=\tabularnewline}

\newcommand{\UMD}{University~of~Maryland,~College~Park,~Maryland~20742,~USA}
\newcommand{\QUICS}{Joint~Center~for~Quantum~Information~and~Computer~Science, NIST/\UMD}
\newcommand{\MITctp}{Center~for~Theoretical~Physics, a~Leinweber~Institute, Massachusetts~Institute~of~Technology, Cambridge,~Massachusetts~02139,~USA}
\newcommand{\IBMykt}{IBM~Quantum,~IBM~T.J.~Watson~Research~Center,~Yorktown~Heights,~New~York~10598,~USA}

\definecolor{FRed}{HTML}{B3242B}
\definecolor{FlowSlate}{HTML}{4F5D68}
\definecolor{FlowSage}{HTML}{628C83}
\definecolor{FlowOchre}{HTML}{B98255}
\definecolor{FlowPlum}{HTML}{786A91}
\definecolor{FlowMist}{HTML}{F3F5F5}

\definecolor{ReviewPink}{HTML}{D6337F}
\newenvironment{review}{\color{ReviewPink}}{}

\begin{document}

\title{Symmetry and AI-assisted discovery of magic-state factories}

\author{Shubham P.\ Jain}
\email{shubhamj810@gmail.com}
\affiliation{\QUICS}
\affiliation{\IBMykt}
\author{Adam Wills}
\email{a_wills@mit.edu}
\affiliation{\IBMykt}
\affiliation{\MITctp}
\author{Shraddha Singh}
\email{shraddha.s@ibm.com}
\affiliation{\IBMykt}

\date{\today}

\begin{abstract}
Magic-state distillation is a major resource cost in fault-tolerant quantum computing.
The cost of a magic-state factory depends strongly on its failure rate, which grows with the number of input magic states.  Although symmetry-restricted
methods have recently made distance two searches tractable,
distance three and above have remained elusive at moderate input counts.  We
develop symmetry- and AI-assisted methods to search this regime.  We present a unified
binary-matrix formulation encompassing both triorthogonal-code and
direct circuit searches.  We show that distance at least three is
equivalent to nonzero, pairwise distinct syndromes, separating the
choice of syndromes from the search for compatible output gates.  We restrict
the syndrome search using group symmetry and language-model agents, followed
by deterministic solving and independent verification.  Our searches yield
$699$ factory classes, including $564$ new ones.  These include factories for
pure-$\Tgate$ states and factories with entangled outputs comprising
combinations of $\Tgate$, $\CS$, and $\CCZ$ magic states.  The pure-$\Tgate$
factories
$\factory{63}{11}{3}$ and $\factory{850}{128}{6}$ achieve the lowest overhead
exponents we know among protocols with at most $100$ and $1000$
inputs, respectively, with $\gamma=1.589$ and $\gamma=1.057$.  Our
$\factory{1715}{287}{6}$ factory, with $\gamma=0.998$, is the smallest known
pure-$\Tgate$ factory with $\gamma<1$.  We also provide a context
directory of search briefs and campaign notes with which readers can train their own agents and tailor the search to their requirements. With these results, we begin constructing an active, open-source repository of magic-state distillation protocols for the quantum community, supplemented by our methods and data, for the practical fault-tolerant quantum computing regime.
\end{abstract}

\maketitle

\section{Introduction}\label{sec:intro}

Fault-tolerant quantum computation protects logical information using quantum
error-correcting codes~\cite{nielsen2010quantume,gottesmansurvivinga,spencer2026quantum}.  On
most practically accessible codes~\cite{gottesman1997stabilizer},
the Clifford gates, generated by the Hadamard, the phase gate
$S=\operatorname{diag}(1,i)$ and the CNOT, can often be implemented by transversal gates or lattice
surgery~\cite{fowler2019low,litinski2019game}.
  But no code
that detects an arbitrary single-qubit error has a universal transversal
gate set~\cite{eastin2009restrictionsc}.  A common way to achieve universality
fault-tolerantly is to supply a non-Clifford gate $U$ through gate
teleportation~\cite{gottesman1999demonstratinga,zhou2000methodology}.  Gate
teleportation is a Clifford circuit by which the \emph{magic state}
$\ket{U}$, obtained by applying $U$ to $\ket{+}$ states, is \emph{injected}
to apply $U$.  The teleported gate is fault-tolerant when the Clifford circuit
is fault-tolerant, which is generally trivial, and the magic state is of high-quality.
Supplying such states is the problem that \emph{magic-state
distillation}~\cite{knill2005quantum,bravyi2005universalb} was designed to
solve.  It converts many low-quality magic states into fewer high-quality
ones.  Resource estimates
attribute a large share of a machine's qubits and time to
distillation~\cite{gidney2025how,zhou2025resource}.

The magic states most studied in the distillation literature are $\ket\Tgate$,
$\ket\CS$ and $\ket\CCZ$.  They teleport the $\Tgate=\operatorname{diag}(1,e^{i\pi/4})$
gate, the $\CS=\operatorname{diag}(1,1,1,i)$ gate and the doubly controlled $Z$
gate $\CCZ$, respectively.  Among these, a noisy $\Tgate$ state is the easiest
to encode in most practical
codes~\cite{li2015magic,singh2022highfidelity,zhang2025constantoverhead,liu2026insitu}.
We therefore concentrate on factories whose inputs are $\Tgate$ states.  For
efficiency of the compiled circuits, however, we would like to analyze as
large a set of outputs as possible.  We restrict to the third level of the
\emph{Clifford hierarchy}~\cite{gottesman1999demonstratinga,cui2017diagonala}.
These are the gates that map Pauli operators to Clifford operators under
conjugation, and their gate teleportation circuit is
trivial.\footnote{States in deeper levels of the Clifford hierarchy require,
with some probability, non-Clifford corrections for gate teleportation to
succeed.}  Importantly, our set of targets spans every magic state required to
teleport a diagonal gate in the third level of the Clifford
hierarchy~\cite{cui2017diagonala}.

The original idea behind a distillation circuit is inspired by QEC codes with transversal implementation of non-Clifford gates. It consumes
$n$ noisy input magic states, each with an independent and identically distributed error of probability $p$, on which $r$ \emph{parity checks} are measured, each the parity of
the faults on some subset of the inputs. The \emph{distance} $d$ is the smallest number of faulty input magic states that pass every check and still corrupt the output, thus yielding an error scaling of $O(p^d)$.  If every single fault is detected, the rejection probability
is approximately $np$ when $np\ll1$.  We write
such a factory as $\factory nkd$, and its
\emph{width} $N=k+r$ counts the logical qubits, or \emph{wires}, it acts on.
For pure-$\Tgate$ outputs the literature uses a standard figure of merit, the
\emph{overhead exponent} $\gamma$~\cite{bravyi2012magic,haah2018codes}.
Reaching an output error $\epsilon$ by concatenating a factory with itself
costs
\begin{equation}
 O\bigl(\log^{\gamma}(1/\epsilon)\bigr),\qquad
 \gamma=\frac{\log(n/k)}{\log d},
 \label{eq:gamma}
\end{equation}
inputs per output.  For a fixed output target and required distance, we
therefore minimize the input count $n$.

The first factory distilled fifteen noisy $\ket\Tgate$ states into one on
the $15$-qubit Reed--Muller code, whose $\Tgate$ is
transversal~\cite{bravyi2005universalb}.  Bravyi and Haah reduced the code
requirement to a combinatorial condition.  A factory is a binary
\emph{triorthogonal} matrix with prescribed single, pair and triple overlap
parities of its rows~\cite{bravyi2012magic}.  Haah and Hastings extended the condition to
$\CS$ and $\CCZ$ outputs and built families by puncturing Reed--Muller
codes~\cite{haah2018codes}.  Richer protocols include Toffoli-state
distillation~\cite{jones2013lowoverhead,eastin2013distilling},
\emph{synthillation}, which directly distills the gate a circuit needs rather than
the $\Tgate$ states it would be built
from~\cite{campbell2017unifieda,campbell2017unifying}, \emph{catalysis}, in
which a magic state assists a conversion and is returned
intact~\cite{campbell2011catalysis,gidney2019efficient}, and codes over
binary extension fields~\cite{gong2026magic}.  Existing approaches include
asymptotic constructions, exhaustive classifications at small sizes, and
searches within specified code families.  Asymptotic constructions reach
constant space overhead at large block lengths~\cite{
wills2025constantoverhead,golowich2025asymptotically,%
nguyen2025good,san-jose2026asymptotically}, and architectural studies estimate the cost of surface-code layouts
and injections~\cite{litinski2019magic,itogawa2025efficient,gidney2024magic,jacoby2025magic,sahay2026foldtransversal,rosenfeld2025magic,gupta2024encoding,hirano2026efficient,li2015magic,singh2022highfidelity,zhang2025constantoverhead,liu2026insitu}.
At small sizes, exhaustive classification is possible. Nezami and Haah
enumerated all triorthogonal codes with $n+k\le38$, which are the
pure-$\Tgate$ factories~\cite{nezami2022classification}.  Our companion
paper extends the enumeration to $n\le54$ and to every third-level
output~\cite{wills2026classification}.  Other searches use
particular code structures. Doubling and code conversion make
transversal-$\Tgate$ codes from transversal-Clifford
codes~\cite{bravyi2015doubleda,jain2025transversala,sullivan2024codeb},
borrowed identities carve a factory from a circuit that acts as the
identity on $\ket+^{\otimes N}$~\cite{singh2026borrowed}, and
satisfiability and integer-programming solvers map compact factories at
small widths~\cite{jacinto2026exploring,baldelli2026constructing}.

We search beyond the classified range for factories with tens to hundreds
of inputs, delivering $\Tgate$ states and entangled states at the third
level of the Clifford hierarchy~\cite{gottesman1999demonstratinga,zhou2000methodology},
at distance three and above.
Efficient searches in the borrowed-identity framework already chart
the distance-two landscape~\cite{singh2026borrowed,singh2026shraggy}, and
cultivation~\cite{sahay2026foldtransversal,rosenfeld2025magic,hirano2026efficient,gidney2024magic}
supplies inputs clean enough that a single distance-two distillation step may
suffice for some applications.  The error rates that large algorithms
demand~\cite{gidney2025how}, however, call for distance three or more, and
there the landscape beyond $n\le54$ is uncharted.

Our method rests on one observation.  Each input flips a set of checks,
which we call its syndrome.  When no two inputs touch the same
wires, distance at least three is equivalent to the syndromes being nonzero
and pairwise distinct~\cite{singh2026borrowed} (see Sec.~\ref{sec:labels}).  We therefore choose the
syndromes first, and we call the set of syndromes a factory uses its
parent.  For each parent, an enumeration of compatible output rows determines the
factories it supports.  Enumerating these rows for every parent gives the
exhaustive search used in our companion paper~\cite{wills2026classification} to
extend the classification of Ref.~\cite{nezami2022classification}.  In the language of triorthogonal codes, the parent spans the $X$ stabilizers and the output rows are the logical operators.  When
this is no longer feasible, we use symmetries to restrict the parent search.  A group
acting on the syndromes turns the choice of a parent into a choice of its orbits.  This is our \emph{symmetry slot ansatz}.  The choice of group or parent to start from remains heuristic, and we delegate this choice to AI agents.

A \emph{large language
model} (LLM) is a neural network trained to continue text and an \emph{agent}
uses an LLM to perform multi-step tasks such as reading solver outcomes and
writing proposals for later runs. An \emph{orchestrator}~\cite{adimulam2026orchestration,roman2026orchestral,song2026paperorchestra}
coordinates several agents working on one task.  Agent-based methods have already had success in quantum error correction.  Cruz-Benito et
al.\ used an LLM-guided evolutionary search to discover bivariate
bicycle codes, a family of quantum low-density parity-check
codes~\cite{cruz-benito2026evolutionary}.  Multi-agent systems have
co-designed codes with transversal diagonal gates~\cite{he2026codesigning}.
Other LLM-driven, reinforcement-learning and evolutionary systems have
found codes, encoders and
encodings~\cite{liu2026largelanguagemodel,yan2026omniqec,%
heitritter2026evolving,nautrup2019optimizing,olle2024simultaneous,%
webster2025engineering}.  Beyond quantum, program search with language
models has produced new constructions in combinatorics and new
algorithms~\cite{romera-paredes2024mathematical,novikov2025alphaevolve}.
Neuro-symbolic systems pair a proposing model with an exact engine that
checks every step~\cite{trinh2024solving}.  In our search, the LLM proposes
candidate structures and a deterministic verifier checks the resulting
circuits.

Throughout, the protocols are \emph{postselected diagonal CSS
protocols}, that is, level-three diagonal rotations with independent $Z$ faults,
checks of pure $X$ type, and rejection on any nonzero outcome.
Every factory we report is an explicit circuit, available in our public
database~\cite{jain2026shubhamj810}.

Section~\ref{sec:ai-landscape} reports the catalog and its highlights.
Section~\ref{sec:matrix} develops the matrix formulation and its two
readings, as a code and as a circuit.
Section~\ref{sec:parent} describes the output search for a fixed parent,
Sec.~\ref{sec:slots} introduces symmetry slots, and Sec.~\ref{sec:ai-methods}
presents the agent workflow.  Section~\ref{sec:outlook} discusses limitations, outlook
and future directions.

\begin{figure*}[t]
\centering
\includegraphics[width=\textwidth]{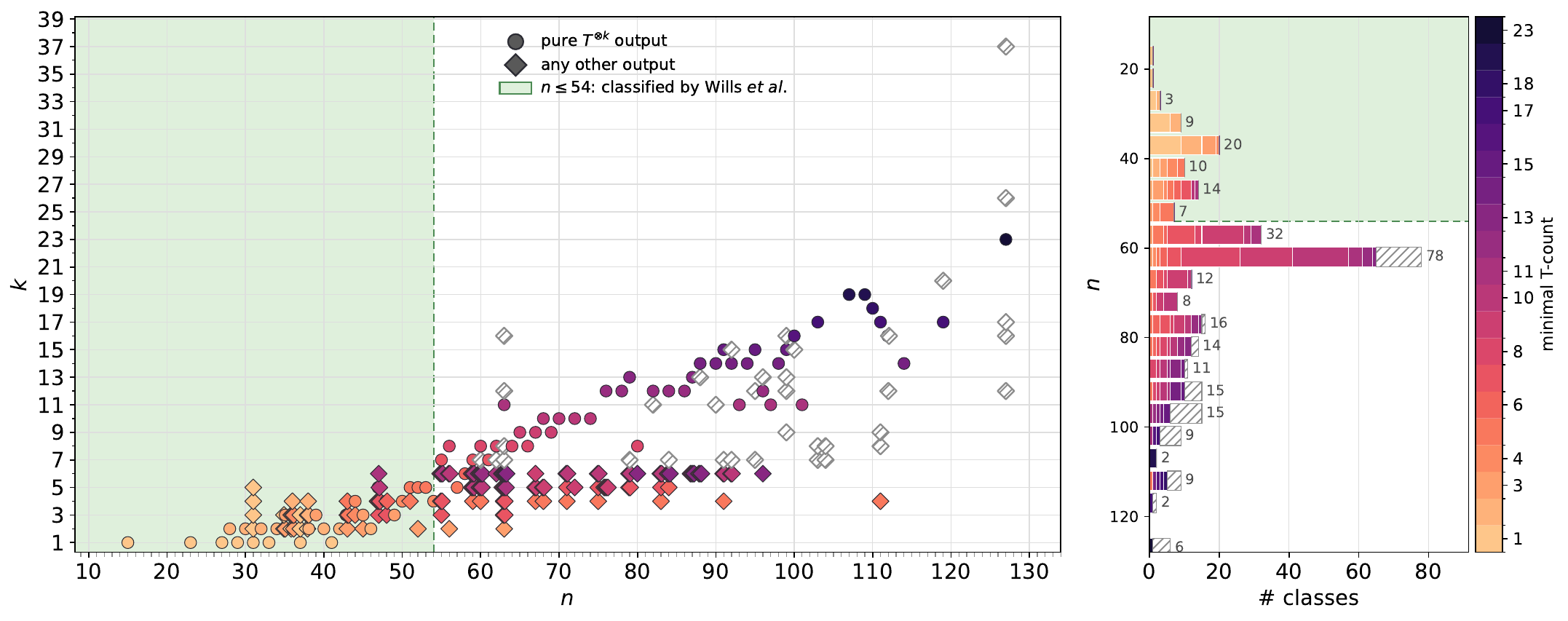}
\caption{The distance-three catalog up to $n=127$.  A
circle is a factory with a pure-$\Tgate^{\otimes k}$ output, and a diamond one with any other output.  Color is the
$\Tgate$-count of the output gate, labeled ``minimal $\Tgate$-count'' on the
color bar. Hatched markers are entangled
outputs too wide for that exact minimization.  A pure-$\Tgate$ factory is
omitted when a wider one at the same $n$ and $d$ already delivers it. The green band is
$n\le54$, the range the companion paper~\cite{wills2026classification}
classifies exhaustively at this distance. The markers inside the band are classes that the methods of this paper found independently of the exhaustive classification. The right panel counts
the drawn classes in bins of five consecutive $n$ across the same range and
splits each bin by $\Tgate$-count.  The bins are placed so that $n=54$ falls on
a boundary, and no bar mixes the two sides of the band.  Of the $294$ classes drawn, $65$ lie in the band,
$225$ are new in this work and the remaining four were already known.}
\label{fig:landscape}
\end{figure*}

\section{Results}\label{sec:ai-landscape}

This section reports what our searches found.  Every
class is a postselected diagonal CSS protocol in the sense of
Sec.~\ref{sec:intro}, written $\factory nkd$ for $n$ inputs, $k$ outputs and
distance $d$.  
 Catalytic and Hadamard-assisted
protocols lie outside this class and could reach parameters the catalog
does not~\cite{singh2026borrowed,gidney2019efficient}, so every comparison
below is within the class.

\begin{table}[t]
\centering\small
\renewcommand{\arraystretch}{1.15}
\setlength{\tabcolsep}{7pt}
\caption{The classes this paper reports, by distance, and how many of them
are new.  A class is a quadruple $(n,k,d,\text{output gate})$ as defined in the
text.  A class counts as new when no published protocol states it and it
lies outside the range of exhaustive classifications.}
\label{tab:bydistance}
\begin{tabular}{@{}l r r r r r r r r@{}}
\toprule
$d$ & $3$ & $4$ & $5$ & $6$ & $7$ & $8$ & $9$ & all\\
\midrule
classes & $493$ & $88$ & $69$ & $42$ & $5$ & $1$ & $1$ & $699$\\
new & $382$ & $72$ & $65$ & $39$ & $4$ & $1$ & $1$ & $564$\\
\bottomrule
\end{tabular}
\end{table}

\begin{table}[t]
\centering\small
\renewcommand{\arraystretch}{1.15}
\setlength{\tabcolsep}{7pt}
\caption{Some highlighted factories resulting from our methods.  The
$\factory{850}{128}{6}$ factory has overhead exponent $\gamma=1.057$
(Table~\ref{tab:rates}), the lowest we know at $n\le1000$.  Each factory
listed has the smallest input count we know of for its output at its distance.}
\label{tab:flagship}
\begin{tabular}{@{}l l r@{}}
\toprule
output gate & factory & $N$\\
\midrule
$\Tgate^{\otimes7}$ & $\factory{55}{7}{3}$ & 14\\
$\Tgate^{\otimes11}$ & $\factory{63}{11}{3}$ & 17\\
$\Tgate^{\otimes16}$ & $\factory{100}{16}{3}$ & 23\\
$\Tgate^{\otimes23}$ & $\factory{127}{23}{3}$ & 30\\
$\Tgate^{\otimes128}$ & {$\factory{850}{128}{6}$} & {166}\\
\midrule
$\CS^{\otimes3}$ & $\factory{63}{6}{3}$ & 12\\
$\CCZ^{\otimes2}$ & $\factory{63}{6}{3}$ & 12\\
$\CCZ^{\otimes4}$ & $\factory{100}{12}{4}$ & 20\\
\midrule
$\Tgate\otimes\CCZ^{\otimes2}$ & $\factory{63}{7}{3}$ & 13\\
$\Tgate^{\otimes9}\otimes\CCZ$ & $\factory{63}{12}{3}$ & 18\\
$\Tgate\otimes\CCZ^{\otimes4}$ & $\factory{99}{13}{3}$ & 20\\
$\Tgate^{\otimes4}\otimes\CCZ^{\otimes4}$ & $\factory{112}{16}{3}$ & 23\\
\bottomrule
\end{tabular}
\end{table}

\begin{table}[t]
\centering\footnotesize
\renewcommand{\arraystretch}{1.15}
\setlength{\tabcolsep}{5pt}
\caption{The punctured Reed--Muller factories tabulated by Haah and
Hastings~\cite{haah2018codes}, beside the lowest overhead exponent we find
at $n\le1000$ for the same distance. Every factory in the table is a pure
$\Tgate^{\otimes k}$ factory, and $\gamma=\log(n/k)/\log d$ is the overhead
exponent, where lower is better. At each distance, our work achieves the lowest exponents known and the distance-six one is the lowest known at any distance at $n\le1000$.}
\label{tab:rates}
\begin{tabular}{@{}c ll ll@{}}
\toprule
 & \multicolumn{2}{c}{Ref.~\cite{haah2018codes}} & \multicolumn{2}{c}{this work}\\
\cmidrule(lr){2-3}\cmidrule(lr){4-5}
$d$ & factory & $\gamma$ & factory & $\gamma$\\
\midrule
$3$ & $\factory{863}{161}{3}$ & $1.528$ & $\factory{750}{210}{3}$ & $\mathbf{1.159}$\\
$4$ & $\factory{872}{152}{4}$ & $1.260$ & $\factory{429}{83}{4}$ & $\mathbf{1.185}$\\
$5$ & $\factory{887}{137}{5}$ & $1.161$ & $\factory{879}{145}{5}$ & $\mathbf{1.120}$\\
$6$ & $\factory{912}{112}{6}$ & $1.170$ & $\factory{850}{128}{6}$ & $\mathbf{1.057}$\\
\bottomrule
\end{tabular}
\end{table}

\begin{figure}[t]
\centering
\includegraphics[width=\columnwidth]{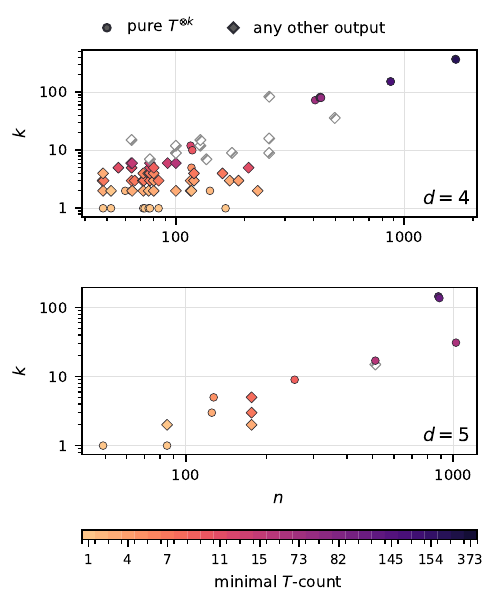}
\caption{The catalog at distances four and five, on logarithmic axes from
$n=48$ to $1683$.  Conventions are as in Fig.~\ref{fig:landscape}, with pure-$\Tgate$ factories at any $n$ represented by their highest $k$ member. Of the $97$ classes drawn, $78$ are new in
this work.}
\label{fig:landscape4}
\end{figure}

Our searches produced $699$ \emph{classes}, a class being a quadruple
$(n,k,d,\text{output gate})$ with the gate taken up to an invertible change of the output basis, the group $\mathrm{GL}(k,2)$,
together with diagonal Clifford corrections (see Appendix~\ref{app:catalogue} for the full list).  In other words, a class represents all the outputs connected by a CNOT+$S$ circuit.  We identify $564$ as new because they lie outside the classified ranges and published protocol parameters.  The new classes span $n=55$
to $1715$, with $320$ at $n\le127$.

To place outputs of different kinds on one scale, we use the minimum unitary
\emph{$\Tgate$-count}.  It is the minimum number of
$\Tgate$ gates required in any circuit of CNOT and $\Tgate$ gates to reach the output gate,
computed exactly by Reed--Muller coset decoding~\cite{amy2019tcount}. We use this computation
up to $k=6$. The $\Tgate$-count of a pure-$\Tgate^{\otimes k}$ output is $k$, a $\CCZ$ output $7$ and a $\CS$ output $3$.  We count over CNOT$+\Tgate$ circuits
because a native output replaces as many separately distilled $\Tgate$
states as its $\Tgate$-count, the states that conventional synthesis would spend on that gate, the
mechanism behind
synthillation~\cite{campbell2017unifieda,campbell2017unifying}.  This metric describes the synthesis cost within CNOT$+\Tgate$ circuits.
It does not determine a factory's full implementation cost. 
Figures~\ref{fig:landscape} and~\ref{fig:landscape4} draw the classes
on the $(n,k)$ plane, colored by this count. Table~\ref{tab:flagship} lists selected factories which have the smallest
input count we know of among postselected diagonal CSS protocols
for its output gate at its stated distance.

Codes with transversal $\Tgate$ gates reach
$\gamma\to0$~\cite{wills2026improved,san-jose2026asymptotically}, but only as
the block length grows, and their performance at the lengths a machine could
afford has not been established.  We therefore compare only explicit
constructions with $n\le1000$.  The natural comparison for the wide
pure-$\Tgate^{\otimes k}$ end of the
catalog is the punctured Reed--Muller triorthogonal family of Haah and
Hastings~\cite{haah2018codes}, whose tables give explicit codes at each
distance from three to seven.  Table~\ref{tab:rates} compares their tabulated
codes at distances three to six with our results at the same distance.
Their best tabulated value, $\gamma=1.161$, is achieved by
$\factory{887}{137}{5}$.  Our searches reach a lower exponent at every distance
in the table, and $\factory{850}{128}{6}$ at $\gamma=1.057$ is the lowest we
know of for an explicit pure-$\Tgate$ factory at $n\le1000$.  We also find a
$\factory{1715}{287}{6}$ factory at $\gamma=0.998$, the smallest known
pure-$\Tgate$ factory with an exponent below one.  Our
entries were found by instantiating the frameworks of
Refs.~\cite{wills2026improved,san-jose2026asymptotically} with the puncture
searches of Appendix~\ref{app:ai-puncture}. 

The improvement is not
confined to the largest codes.  At distance three, the $\factory{127}{23}{3}$
factory uses approximately $5.52$ inputs per output, compared with $5.74$ for
their $\factory{109}{19}{3}$.  Our $\factory{63}{11}{3}$ has a similar ratio,
$5.73$ inputs per output, while consuming $63$ inputs rather than $109$.  Its
exponent, $\gamma=1.589$, is the lowest we know among pure-$\Tgate$ factories
with $n\le100$.


\section{Factory as a binary matrix}\label{sec:matrix}

We represent a magic-state factory by a binary matrix whose columns specify
its parity rotations.  The matrix can then be read in two ways.  Read as a
code, it is a triorthogonal code or one of its
generalizations~\cite{bravyi2012magic,haah2018codes}.  Read as a circuit, it
is the list of parity rotations of a circuit
construction~\cite{jacinto2026exploring,singh2026borrowed}.  Each
reading takes a different object as the unknown.  We
do not commit to either reading and work with the matrix itself.  We choose
its check part first and its outputs second, the syndrome-first formulation
that Sec.~\ref{sec:matrix-problem} introduces. 

\subsection{The matrix and the design problem}\label{sec:matrix-problem}\label{sec:labels}

Fix $N=k+r$ wires, $k$ outputs and $r$ checks, all prepared in $\ket+$.  Each
nonzero $\alpha\in\F_2^N$, a length-$N$ vector over the binary field
$\F_2=\{0,1\}$, names a parity rotation
\begin{equation}
 Z_\alpha=\exp\!\Bigl[i\tfrac\pi8\bigl(I-\textstyle\prod_{a:\alpha_a=1}Z_a\bigr)\Bigr],
 \label{eq:rotation}
\end{equation}
the weight-$w$ phase rotation $Z^{\otimes w}(\pi/8)$ of
Refs.~\cite{jones2013lowoverhead,singh2026borrowed} with $w=\wt(\alpha)$, the Hamming weight of $\alpha$.  It is
a $\Tgate$ rotation of the parity of the wires in the support of $\alpha$.  It costs one input, because a ladder of CNOTs moves a single-qubit $\Tgate$ onto any parity (Fig.~\ref{fig:factory}, legend).  If that input is faulty, the same $\alpha$ is the $Z$ fault it leaves behind.  A \emph{reduced circuit}
is an $N\times n$ binary matrix $G$ with distinct nonzero columns, one per
input. Repeated columns cancel up to a
Clifford, since two $\Tgate$ rotations of the same parity compose to a
Clifford rotation. Column $m$ marks the wires that rotation $m$ touches, and row
operations are CNOTs.  We call a factory \emph{full rank} when its check rows
are independent and its output rows are independent modulo the check span,
equivalently when $\rank G=N$.  When the columns are the unknown we write the same matrix as
$G_y$, where the selection $y\in\F_2^{2^N-1}$ has one bit per nonzero $\alpha$ and records which of the $2^N-1$ possible columns appear.

\begin{figure*}[t]
\centering
\begin{tikzpicture}[x=1cm,y=1cm,
  wire/.style={thick,FlowSlate},
  lab/.style={font=\scriptsize,anchor=east,inner sep=1.5pt},
  dot/.style={circle,fill=FlowSlate,inner sep=0pt,minimum size=4.2pt},
  tbox/.style={draw=FlowOchre,fill=FlowOchre!14,font=\scriptsize\bfseries,
               inner sep=1.6pt,minimum size=0.34cm,rounded corners=1pt},
  bit/.style={font=\small},
  meter/.style={draw=FlowOchre,fill=FlowOchre!12,inner sep=2pt,
                rounded corners=1pt},
  tag/.style={font=\scriptsize\itshape,FlowSlate},
  cnot/.style={circle,draw=FlowSlate,thick,inner sep=0pt,minimum size=7pt}]
\def\xs{1.05}  
\def\dx{0.82}  
\foreach \y/\l in {2.7/1,2.1/2,1.5/3} {
  \draw[wire] (0.3,\y)--(7.55,\y); \node[lab] at (0.3,\y) {$\ket+$};}
\draw[decorate,decoration={brace,amplitude=4pt},FlowSlate,semithick]
  (7.62,2.82)--(7.62,1.38);
\node[anchor=west,font=\small] at (7.78,2.1) {$\ket{\CCZ}$};
\node[tag,anchor=west] at (7.78,1.75) {$k=3$ outputs};
\draw[wire] (0.3,0.7)--(7.15,0.7); \node[lab] at (0.3,0.7) {$\ket+$};
\node[meter,minimum width=0.46cm,minimum height=0.38cm] (mt) at (7.4,0.7) {};
\draw[FlowOchre,semithick] ($(mt.center)+(-0.14,-0.08)$)
  arc[start angle=180,end angle=0,radius=0.14];
\draw[FlowOchre,semithick,-{Latex[length=1mm]}] ($(mt.center)+(0,-0.08)$)
  -- ($(mt.center)+(0.1,0.12)$);
\node[font=\scriptsize,FlowOchre,anchor=south,inner sep=0.5pt] at (mt.north) {$X$};
\node[tag,anchor=west] at (7.7,0.7) {check};
\foreach \m/\a/\b/\c in {0/0/0/0,1/0/0/1,2/0/1/0,3/0/1/1,4/1/0/0,5/1/0/1,6/1/1/0,7/1/1/1} {
  \pgfmathsetmacro{\x}{\xs+\m*\dx}
  \ifnum\a=1 \draw[wire] (\x,2.7)--(\x,0.7); \node[dot] at (\x,2.7) {};
  \else \ifnum\b=1 \draw[wire] (\x,2.1)--(\x,0.7);
  \else \ifnum\c=1 \draw[wire] (\x,1.5)--(\x,0.7); \fi\fi\fi
  \ifnum\b=1 \node[dot] at (\x,2.1) {}; \fi
  \ifnum\c=1 \node[dot] at (\x,1.5) {}; \fi
  \node[tbox] at (\x,0.7) {$\Tgate$};
  \node[bit] at (\x,-0.25) {$\a$};
  \node[bit] at (\x,-0.65) {$\b$};
  \node[bit] at (\x,-1.05) {$\c$};
  \node[bit] at (\x,-1.6) {$1$};
}
\draw[FlowSlate,semithick] (0.62,0.02)--(0.5,0.02)--(0.5,-1.85)--(0.62,-1.85);
\draw[FlowSlate,semithick] (7.1,0.02)--(7.22,0.02)--(7.22,-1.85)--(7.1,-1.85);
\draw[FlowSlate!60] (0.62,-1.32)--(7.1,-1.32);
\node[font=\small,anchor=east] at (0.42,-0.9) {$G=$};
\node[tag,anchor=west] at (7.3,-0.25) {output $1$};
\node[tag,anchor=west] at (7.3,-0.65) {output $2$};
\node[tag,anchor=west] at (7.3,-1.05) {output $3$};
\node[tag,anchor=west] at (7.3,-1.6) {check};
\begin{scope}[shift={(9.95,0)}]
\node[font=\scriptsize,anchor=west] at (-0.1,3.05) {parity rotation $Z_\alpha$ on wires $\{a,b,c\}$:};
\foreach \y in {2.6,2.1,1.6} {\draw[wire] (0.1,\y)--(0.85,\y); \draw[wire] (1.25,\y)--(3.9,\y);}
\node[lab] at (0.1,2.6) {$a$}; \node[lab] at (0.1,2.1) {$b$}; \node[lab] at (0.1,1.6) {$c$};
\draw[wire] (0.55,2.6)--(0.55,1.6); \node[dot] at (0.55,2.6) {}; \node[dot] at (0.55,2.1) {};
\node[tbox] at (0.55,1.6) {$\Tgate$};
\node[font=\small] at (1.05,2.1) {$=$};
\draw[wire] (1.5,2.6)--(1.5,1.6); \node[dot] at (1.5,2.6) {}; \node[cnot] at (1.5,1.6) {};
\draw[wire] (1.5,1.72)--(1.5,1.48);
\draw[wire] (2.0,2.1)--(2.0,1.6); \node[dot] at (2.0,2.1) {}; \node[cnot] at (2.0,1.6) {};
\draw[wire] (2.0,1.72)--(2.0,1.48);
\node[tbox] at (2.5,1.6) {$\Tgate$};
\draw[wire] (3.0,2.1)--(3.0,1.6); \node[dot] at (3.0,2.1) {}; \node[cnot] at (3.0,1.6) {};
\draw[wire] (3.0,1.72)--(3.0,1.48);
\draw[wire] (3.5,2.6)--(3.5,1.6); \node[dot] at (3.5,2.6) {}; \node[cnot] at (3.5,1.6) {};
\draw[wire] (3.5,1.72)--(3.5,1.48);
\node[font=\scriptsize,anchor=west,text width=4.1cm] at (-0.1,0.55)
  {one $\Tgate$ input per rotation; the CNOT ladder moves the $\Tgate$
   onto the parity of the marked wires, and a faulty input is a
   $Z$ error on exactly those wires};
\end{scope}
\end{tikzpicture}
\caption{The $\factory{8}{3}{2}$ $\CCZ$ factory as a circuit~\cite{singh2026borrowed} and as a matrix~\cite{kubica2015unfolding,%
2016smallest}.  Four wires prepared in $\ket+$
receive eight parity rotations $Z_\alpha$, drawn as a vertical connector with
a dot on each wire in the support of $\alpha$ and the $\Tgate$ box on the
check wire. Beneath each rotation is its corresponding column in the matrix $G$ with support on every
wire the rotation touches, the top three rows being the outputs and the
bottom row the check. At the end the check wire is measured in the $X$ basis and the
run is kept only if the outcome is $+1$ (no check flipped). Three output wires carry the delivered 
$\ket{\CCZ}=\CCZ\ket{+\!+\!+}$ state.  By Eq.~\eqref{eq:em-gate}, every column count is even except that of the
triple $\{1,2,3\}$, which lies only in $(111\mid1)$, so the circuit
applies $\CCZ_{123}$.  A single
fault flips the check and is detectable. Two faults cancel on the check but, being distinct
columns, differ on the outputs, so $d=2$~\cite{singh2026borrowed}.}
\label{fig:factory}
\end{figure*}
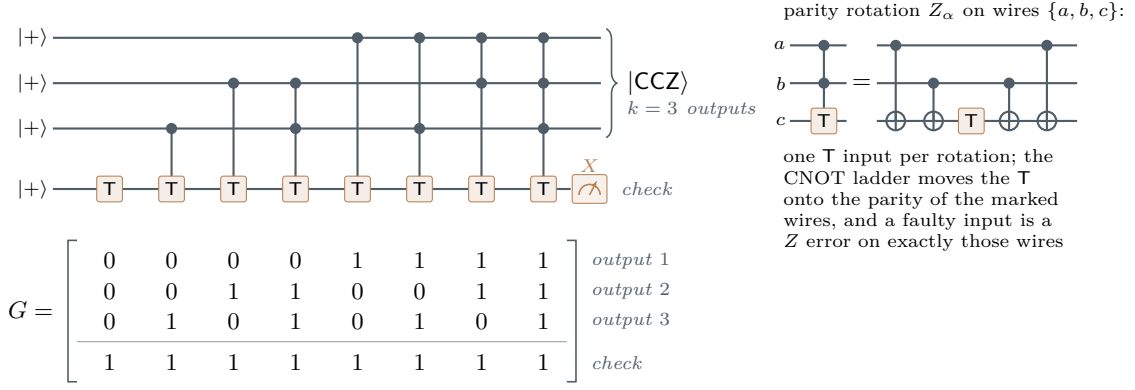

The gate the circuit implements can be read off by counting columns.  A
parity rotation $Z_\alpha$ multiplies a computational basis state $\ket z$
by $e^{i\pi/4}$ raised to the parity $\bigoplus_{a:\alpha_a=1}z_a$.  That parity obeys the integer identity
$\bigoplus_az_a\equiv\sum_az_a-2\sum_{a<b}z_az_b+4\sum_{a<b<c}z_az_bz_c
\pmod 8$.  Applying all $n$ rotations
and collecting terms, the circuit acts, up to a diagonal Clifford, as
\begin{equation}
 \prod_a\Tgate_a^{\,c_a}\ \prod_{a<b}\CS_{ab}^{\,c_{ab}}\
 \prod_{a<b<c}\CCZ_{abc}^{\,c_{abc}},
 \label{eq:em-gate}
\end{equation}
where the \emph{column counts} $c_a$, $c_{ab}$ and $c_{abc}$ are the numbers
of columns that contain wire $a$, both $a$ and $b$, and all three of $a$, $b$
and $c$.  The sign of the middle term is absorbed into the diagonal Clifford,
since $\CS^{-1}$ differs from $\CS$ by a controlled-$Z$.  Since $\Tgate^2$,
$\CS^2$ and $\CCZ^2$ are Clifford gates, only the parities of these counts
carry non-Clifford content.  A wire, pair or
triple lying in an odd number of columns receives a $\Tgate$, a $\CS$ or a
$\CCZ$ respectively, and the rest is a diagonal Clifford that is tracked
classically.  In the matrix $G$, wire $a$ is row $g_a$, and the count $c_Q$
of columns containing the wires in a set $Q$ is the Hamming weight of
the coordinatewise product $\bigodot_{a\in Q}g_a$, so the gate is fixed by
the parities
\begin{equation}
 \wt\Bigl(\,\bigodot_{a\in Q}g_a\Bigr)\bmod 2,\qquad 1\le|Q|\le3 .
 \label{eq:em-overlaps}
\end{equation}
We call this parity the \emph{moment} of $Q$, and $|Q|$ its \emph{degree}, so
degrees one, two and three carry the $\Tgate$, $\CS$ and $\CCZ$ content.
Figure~\ref{fig:factory} carries out this count for a worked example, the
$\factory{8}{3}{2}$ $\CCZ$ factory.

The gate is thus a diagonal gate in the
third level of the Clifford hierarchy, given as a product of integer
powers of $\Tgate$, $\CS$ and $\CCZ$ and determined up to a tracked
diagonal Clifford gate.  The target vector $D$ lists the parities the
output gate requires, one bit for each index $Q$ with $|Q|\leq 3$. Every entry whose $Q$ touches a check wire is zero, because a
check must carry no non-Clifford phase or entanglement with the output wires.  The monomial
matrix $R_N$, with entries $(R_N)_{Q,\alpha}=\prod_{a\in Q}\alpha_a$, has
one row per index $Q$ and one column per nonzero $\alpha$, and its entry
is one exactly when $\alpha$ contains $Q$.  The product $R_Ny$
therefore counts, modulo two, the selected columns that contain each $Q$. This is 
exactly the parity of Eq.~\eqref{eq:em-overlaps}.  For the factory of
Fig.~\ref{fig:factory}, the target vector has a single one, at $Q=\{1,2,3\}$, with the rest being zero.  

A selection $y$ is a valid factory for the target $D$ when $R_Ny=D$. A valid factory still has to protect its outputs, and the distance $d$ of Sec.~\ref{sec:intro}, which we now write $d(G)$, measures how well it does.  Split each column into its output and check parts,
$\alpha^m=(P_m\mid S_m)\in\F_2^k\times\F_2^r$. $S_m$ is the \emph{syndrome} of
fault $m$, the checks it flips, and $P_m$ is its \emph{label}, the outputs it corrupts.  A set
$F$ of faulty inputs is \emph{accepted} when its syndromes cancel and
\emph{harmful} when its labels do not. The circuit distance is then
\begin{equation}
 d(G)=\min_F\bigl\{|F|:\ \textstyle\bigoplus_{m\in F}S_m=0,\
 \bigoplus_{m\in F}P_m\ne0\bigr\}.
 \label{eq:distance}
\end{equation}
the size of the smallest fault set that is both accepted and harmful.
The factory design problem is then a search for $y$ such that
\begin{equation}
 R_Ny=D\qquad \textrm{and} \qquad d(G_y)\ge d.
 \label{eq:em-problem}
\end{equation}
Minimizing $n=\wt(y)$ optimizes the cost.

The distance condition is simpler than it looks.  For a reduced circuit,
Eq.~\eqref{eq:distance} gives $d(G)\ge2$ if and only if every syndrome
$S_m$ is nonzero, since a zero syndrome makes its single fault accepted
and, by reducedness, harmful.  The distance-three condition is the observation this work rests on.

\begin{theorem}[distance-three syndrome criterion]\label{thm:syndrome-criterion}
For a reduced circuit $G$, $d(G)\ge3$ if and only if its syndromes are
nonzero and pairwise distinct.
\end{theorem}

\begin{proof}
Nonzero syndromes exclude single faults.  A pair $i,j$ is accepted iff
$S_i=S_j$, and it is then harmless only if $P_i=P_j$, which would make the
two columns identical, contrary to reducedness.
\end{proof}

Distance three is therefore a condition on the syndromes alone, and a
distance-$\ge3$ circuit is specified by
\begin{equation}
 \underbrace{\cS\subseteq\F_2^r\setminus\{0\}}_{\text{which syndromes occur}}
 \quad\text{and}\quad
 \underbrace{\lambda:\cS\to\F_2^k}_{\text{what each one outputs}} ,
 \label{eq:split}
\end{equation}
where the \emph{parent} $\cS$, the set of
syndromes that occur, fixes what is detected, and the \emph{labeling} $\lambda$,
with $\lambda(S_m)=P_m$, fixes what is delivered. Since $n=|\cS|$, a
factory with $r$ checks has at most $2^r-1$ inputs.  Higher distances are conditions of the same kind.  For instance, $d\ge4$ iff
$\lambda$ is additive on every triple of syndromes summing to zero.  Writing a
factory this way enforces distance three without enumerating fault sets, which
is what makes our searches affordable.

\subsection{Reading it as a code}\label{sec:matrix-code}

Split the rows of $G_y$ into the output block $P$, whose columns are the
labels $P_m$, and the check block $S$, whose columns are the syndromes
$S_m$ and whose column set is the parent $\cS$.  Then $G_y=\binom PS$ is
the generator data of a CSS code.  Its $X$ stabilizers are the row space of
$S$ and its $Z$ stabilizers are $G_y^{\perp}$, the vectors orthogonal to every
row of $G_y$ under the binary inner product $\langle u,v\rangle=\bigoplus_iu_iv_i$.
The two commute because $\Span S\subseteq\Span G_y$.  If $\rank S=r$ and
$\rank G_y=r+k$, the code encodes $k$ logical qubits.  Row by row, $R_Ny=D$
prescribes the moments of Eq.~\eqref{eq:em-overlaps}.  For the target
$\Tgate^{\otimes k}$ this is triorthogonality~\cite{bravyi2012magic}, and for other diagonal and
entangled outputs its
generalization~\cite{haah2018codes,campbell2017unifieda,%
campbell2017unifying}.

The distance of Eq.~\eqref{eq:distance} is the $Z$-distance of this code:
a $Z$ operator with support $e\in\F_2^n$ commutes with every $X$ check when
$Se=0$ and is a stabilizer when also $Pe=0$, so
$d_Z=\min\{\wt(e):Se=0,\ Pe\ne0\}$, which is Eq.~\eqref{eq:distance} with
$e$ for $F$.  Only the $Z$-distance matters because input faults are
$Z$ faults and the inner code handles the rest.  This reading connects
factories to weight enumerators, puncturing, concatenation and the
classification of small triorthogonal codes~\cite{nezami2022classification}.

\subsection{Reading it as a circuit}\label{sec:matrix-identity}\label{sec:matrix-search}

Section~\ref{sec:matrix-problem} already builds $G$ from a circuit, and
Fig.~\ref{fig:factory} works through the $\factory{8}{3}{2}$ example.  Each
row is a wire, each column is an input and the fault it can leave. A
selection $y$ implements the target when $R_Ny=D$.  The circuit reading
therefore adds no new objects.  It offers two perspectives on the same circuit,
which lead to two ways of searching for a solution of
Eq.~\eqref{eq:em-problem}.  A \emph{direct circuit search} looks for the
selection $y$ itself.  A \emph{borrowed-identity search} starts from a circuit
that acts as the identity and deletes columns from it.

The first perspective treats Eq.~\eqref{eq:em-problem} as a
\emph{satisfiability} (SAT) problem, a set of Boolean variables together with
constraints they must satisfy jointly, as in the circuit search of Jacinto et
al.~\cite{jacinto2026exploring} and the solver searches for triorthogonal
codes~\cite{baldelli2026constructing}.  The variables are the bits of the
selection $y$, one per candidate column.  The target enters as linear
equations in $y$, and the distance as clauses that forbid each low-weight
fault set that is accepted and harmful.  The cost lies in the distance
clauses, which range over fault subsets, and
Theorem~\ref{thm:syndrome-criterion} removes that cost up to $d=3$.

The second perspective starts from circuits whose overall action is a
diagonal Clifford.  These are the selections $h$ with $R_Nh=0$, which form the kernel
$\cK_N=\ker R_N$.  A selection $h\in\cK_N$ is an identity up to a diagonal
Clifford, and \emph{borrowed identities}~\cite{singh2026borrowed} lie in this
kernel.  Any two solutions $y,y'$ of $R_Ny=D$ differ by an element of
$\cK_N$, so the solution set is one \emph{coset} of the kernel, a copy of
$\cK_N$ shifted by a single solution $y_0$,
\begin{equation}
 \{y:R_Ny=D\}=y_0+\cK_N .
 \label{eq:em-coset}
\end{equation}
We call $h=y+y_0$ the \emph{displacement} of a solution $y$ from $y_0$.  The
minimum input count for a target is the weight of a \emph{coset leader},
the lightest element of that coset.  Finding one is the decoding problem that
Amy and Mosca identified with $\Tgate$-count
minimization~\cite{amy2019tcount}.

An identity becomes a factory by deletion.  Pick a bare circuit $b_D$ that
implements the target on the $k$ output wires alone, so that every column of
$b_D$ is \emph{output-only}, with no support on a check wire.  For any $y_0$
with $R_Ny_0=D$, the sum $h=y_0+b_D$ lies in $\cK_N$.  Conversely, toggling
the output-only columns $b_D$ out of an identity $h$ leaves
\begin{equation}
 y_0=h+b_D ,
 \label{eq:em-toggle}
\end{equation}
an encoded factory for $D$.  This is the deletion step of
Ref.~\cite{singh2026borrowed}.  The surviving columns all touch a check, so
the distance is at least two by construction.  The output is only as rich as
the output-only columns the identity had to delete.  A strict borrowed
identity in the sense of Ref.~\cite{singh2026borrowed} is a global phase on
every basis state.  In the column counts $c_Q$ of Eq.~\eqref{eq:em-gate},
this requires
\begin{align}
 c_a &\equiv 0 \pmod 8 \quad\text{for every wire } a,\label{eq:em-literal}\\
 c_{ab} &\equiv 0 \pmod 4 \quad\text{for every pair } a<b,\\
 c_{abc} &\equiv 0 \pmod 2 \quad\text{for every triple } a<b<c .
\end{align}
Equation~\eqref{eq:em-toggle} needs only $R_Nh=0$, which permits a
leftover diagonal Clifford and is therefore strictly more general.

The code reading and both perspectives of the circuit reading describe one
object, and the following theorem says so.  Item (i) is the circuit reading
with the selection $y$ as coordinates, item (ii) is the code reading, and
item (iii) is the circuit reading with the displacement $h$ as coordinates.

\begin{theorem}[one coset, two readings]\label{thm:equivalence}
For selections $y,y'$ and any attainable level-three diagonal target $D$:
\begin{enumerate}[label=(\roman*),leftmargin=2.2em,itemsep=1pt]
\item $y$ and $y'$ implement the same gate up to a diagonal Clifford and global
phase if and only if $R_Ny=R_Ny'$.
\item For every $Q\subseteq\{1,\dots,N\}$ with $1\le|Q|\le3$,
 $(R_Ny)_Q=\wt(\bigodot_{a\in Q}g_a)\bmod2$. Hence $G_y$ is the generalized
 triorthogonal matrix whose prescribed moments are $D$.
\item For any $y_0$ with $R_Ny_0=D$, the map $h\mapsto y_0+h$ is a bijection
 from $\cK_N$ onto $\{y:R_Ny=D\}$.
\end{enumerate}
\end{theorem}

\begin{proof}
By Eq.~\eqref{eq:em-gate}, two selections differ by a diagonal Clifford
exactly when their column-count parities, the entries of $R_Ny$ and
$R_Ny'$, agree, proving (i).  For (ii), $\bigodot_{a\in Q}g_a$ is one in column $m$ exactly when
$\alpha^m_a=1$ for all $a\in Q$, and its parity is the $Q$th entry of $R_Ny$.
For (iii), $R_N(y_0+h)=D$ iff $R_Nh=0$, which is the definition of $\cK_N$.
\end{proof}

\begin{table}[!htb]
\centering\footnotesize
\renewcommand{\arraystretch}{1.2}
\caption{The binary matrix formulation underlies both readings of a
magic-state factory.  Read as a code, the matrix is a triorthogonal code.  Read
as a circuit, it is searched through a borrowed identity or directly.  
}
\label{tab:readings}
\setlength{\tabcolsep}{7pt}\begin{tabular}{@{}lll@{}}
\toprule
\rr reading & \rr perspective & \rr search object \\
\midrule
\rr code & \rr triorthogonal code & \rr rows of $G_y$ \\
\midrule
\rr circuit & \rr borrowed identity & \rr displacement $h\in\cK_N$ \\
\rr & \rr direct circuit search & \rr selection $y$ \\

\bottomrule
\end{tabular}
\end{table}

A factory found by deleting from an identity and the ``same'' factory
found by a triorthogonal solver are therefore the same point of the coset
$y_0+\cK_N$ of Eq.~\eqref{eq:em-coset}.  The two readings, and the two
perspectives within the circuit reading, use different coordinates on that
coset (Table~\ref{tab:readings}).  Imposing distance at least three makes
available an intermediate between the two readings, the syndrome-first
formulation of Eq.~\eqref{eq:split}, which we use for the searches below.

\section{Parent evaluation}\label{sec:parent}

For a fixed parent, the remaining variables are the
output rows.  The conditions for compatibility with the checks are linear
or bilinear over
$\F_2$.  An output row must not leak a non-Clifford phase onto a check, two
output rows may coexist only if their coordinatewise product is orthogonal
to every check, and two rows that differ by a check row are the same
output.
The first and third conditions make the legal output directions the
quotient
\begin{equation}
 V(S)=\bigl(S^{\langle2\rangle}\bigr)^{\!\perp}\big/\,S ,
 \label{eq:menu}
\end{equation}
and the second is the \emph{compatibility condition}
\begin{equation}
 \langle a\od b,\,s\rangle=0\qquad\text{for every }s\in S
 \label{eq:compat}
\end{equation}
on any two output representatives $a,b$. Here, $S^{\langle2\rangle}=\Span\{s\od s':s,s'\in S\}$ is the
Schur square of $S$, the span of all coordinatewise products of
check rows. Since $s\od s=s$ over $\F_2$ it contains $S$.  The target
condition of Sec.~\ref{sec:matrix-problem} makes every overlap that
involves a check row even, so $S\subseteq(S^{\langle2\rangle})^{\perp}$. The dimension of $V(S)$ is the parent's
\emph{capacity} $\kappa$, given by
\begin{equation}
 \kappa=\dim V(S)\ \ge\ n-2r-\binom r2.
 \label{eq:capacity}
\end{equation}
Since $S^{\langle2\rangle}$ is spanned by the $r$ check rows and their
$\binom r2$ pairwise products, its orthogonal complement has dimension
at least $n-r-\binom r2$. The quotient by $S$ removes at most $r$ more.
Within the space $V(S)$, one chooses a valid frame, that is, $k$ elements $a_1,\dots,a_k$ of $V(S)$ that
are independent and satisfy Eq.~\eqref{eq:compat} in pairs. The output gate is then read off from the parities of the chosen frame via
Eq.~\eqref{eq:em-overlaps}.

For a fixed parent, we compute $V(S)$ and enumerate compatible independent
output frames. Each frame determines an output gate through its overlap
parities. The dimension $\kappa$ controls the size of the frame space.  A completed
enumeration gives all compatible gates for that parent, including product,
controlled-phase and entangled outputs, without specifying a target gate in
advance.  For example, an incomplete enumeration of the parent that consists of all
nonzero columns of length $6$ already found $418,281$ distance-three factories at
$n=63$, inequivalent up to a change of the output basis, the group
$\mathrm{GL}(k,2)$. When full enumeration is infeasible, we can instead search for a
specified target.  Given a partial frame, the overlap conditions on the next
output row form an affine system.  Such a solve tests extensions of that
partial frame.  Ruling out the target requires exhausting the relevant partial
frames.



\section{Symmetry slots}\label{sec:slots}

The parent evaluation of Sec.~\ref{sec:parent} searches for compatible
outputs once the syndromes are fixed.  The remaining challenge is
to find promising parents: with $r$ checks, there are $2^{2^r-1}$ subsets of
nonzero syndromes to consider, already $2^{127}$ at $r=7$.  We use group
symmetry to organize this search into families that can be explored
efficiently.

Our \emph{symmetry slot ansatz} groups syndromes into orbits under a chosen
group action.  We require the parent to be a union of these orbits and assign
a common output label to all syndromes in each selected orbit.  The solver
then chooses the orbits and their labels jointly to realize the target gate.
This replaces decisions about individual syndromes with decisions about
entire orbits, reducing the search space while allowing both the parent and
its output labeling to be determined within the chosen symmetry.

Let $\cG\le\mathrm{GL}(r,2)$, the group of invertible
$r\times r$ binary matrices, act on the nonzero syndromes. Let $\Omega_\cG$
be the set of orbits, the equivalence classes of syndromes under $\cG$,
with $q_\cG=|\Omega_\cG|$.  A \emph{slot} is a pair $(P,\omega)$ with
$\omega\in\Omega_\cG$ and $P\in\F_2^k$.  Selecting the slot $(P,\omega)$ places in $G$ one
column $(P\mid s)$ for each $s\in\omega$, so all of them carry the same label
$P$.  Each orbit may be assigned at most one output label.  A circuit matrix $G$ is then specified by a
selection of slots whose columns together realize the target gate, $R_Ny=D$.
This is Eq.~\eqref{eq:split} with orbits in place of syndromes, and the input
count is
\begin{equation}
 n=\sum_{(P,\omega)\ \text{selected}}|\omega| .
 \label{eq:slotcount}
\end{equation}
Figure~\ref{fig:slots} illustrates the ansatz for a representation of
the cyclic group $C_3$.

\begin{figure}[t]
\centering
\begin{tikzpicture}[x=1cm,y=0.9cm,
  wire/.style={circle,draw=FlowSlate,fill=FlowMist,thick,minimum size=0.36cm,
               inner sep=0pt,font=\scriptsize},
  rot/.style={-{Latex[length=1.2mm]},FlowPlum,thick},
  tag/.style={font=\scriptsize\itshape,FlowSlate},
  panel/.style={tag,black},
  orb/.style={draw=FlowSlate!60,fill=FlowMist,rounded corners=3pt,inner sep=3pt,font=\scriptsize},
  sel/.style={draw=FlowOchre,fill=FlowOchre!10,rounded corners=3pt,inner sep=3pt,thick,font=\scriptsize}]
\node[panel,anchor=west] at (0,2.2) {(a) $\cG$ rotates two triples of check wires and fixes the seventh};
\foreach \i in {1,...,7} {\pgfmathsetmacro{\x}{0.45+(\i-1)*0.62+(\i>3)*0.3+(\i>6)*0.3}
  \node[wire] (c\i) at (\x,1.55) {$\i$};}
\draw[rot] (c1) to[bend left=50] (c2); \draw[rot] (c2) to[bend left=50] (c3);
\draw[rot] (c3) to[bend left=45] (c1);
\draw[rot] (c4) to[bend left=50] (c5); \draw[rot] (c5) to[bend left=50] (c6);
\draw[rot] (c6) to[bend left=45] (c4);
\node[tag,anchor=west] at (5.1,1.55) {fixed};
\node[panel,anchor=west] at (0,0.75) {(b) orbits of syndromes under $\cG$};
\node[orb,anchor=west] (o1) at (0.2,0.2) {\texttt{0000001}};
\node[tag,anchor=west] at (1.55,0.2) {fixed point, size $1$ ($7$ such)};
\node[sel,anchor=west] (o2) at (0.2,-0.5)
  {\texttt{1001000} $\xrightarrow{\,\cG\,}$ \texttt{0100100} $\xrightarrow{\,\cG\,}$ \texttt{0010010}};
\node[tag,anchor=west] at (5.3,-0.5) {$\omega$, size $3$ ($40$ such)};
\node[panel,anchor=west] at (0,-1.15) {(c) the slot $(P,\omega)$ with $P=0110$: three columns of $G$};
\foreach \r/\a/\b/\c in {1/0/0/0,2/1/1/1,3/1/1/1,4/0/0/0,5/1/0/0,6/0/1/0,7/0/0/1,8/1/0/0,9/0/1/0,10/0/0/1,11/0/0/0} {
  \pgfmathsetmacro{\y}{-1.7-(\r-1)*0.29-(\r>4)*0.08}
  \node[font=\scriptsize] at (0.65,\y) {$\a$}; \node[font=\scriptsize] at (0.95,\y) {$\b$}; \node[font=\scriptsize] at (1.25,\y) {$\c$};}
\draw[FlowSlate!70] (0.5,-2.755)--(1.4,-2.755);
\draw[FlowOchre,thick,rounded corners=3pt] (0.45,-1.5) rectangle (1.45,-4.85);
\node[tag,anchor=west] at (1.75,-2.15) {$k=4$ output rows with label $P$ in every column};
\node[tag,anchor=west,align=left] at (1.75,-3.4)
  {$r=7$ check rows: the syndromes of $\omega$};
\node[tag,anchor=west] at (1.75,-4.4) {$|\omega|=3$ inputs with at most one label per orbit};
\end{tikzpicture}
\caption{The slot ansatz, drawn for one example group.  (a) $\cG$ is the
order-three permutation of the $r=7$ check wires that rotates
$\{1,2,3\}$ and $\{4,5,6\}$ and fixes wire $7$.  (b) Syndromes fall into orbits under $\cG$.  The seven
syndromes constant on each triple are fixed points, and the other $120$
form $40$ orbits of size three, $47$ orbits in all.  (c) A slot pairs an
orbit $\omega$ with an output label $P\in\F_2^k$ and stands for the
columns $(P\mid s)$ of $G$, one for each $s\in\omega$. 
}
\label{fig:slots}
\end{figure}
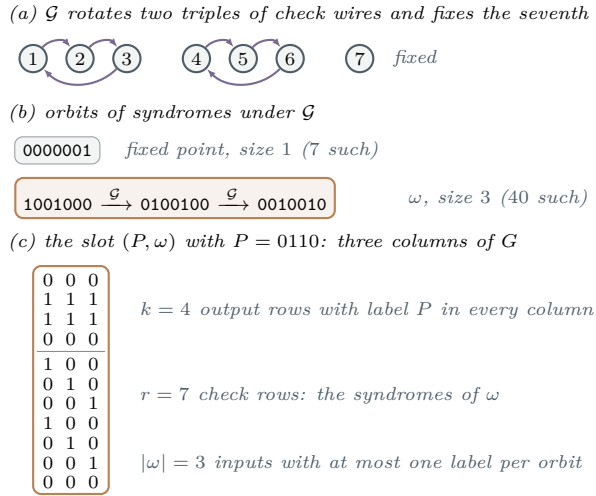

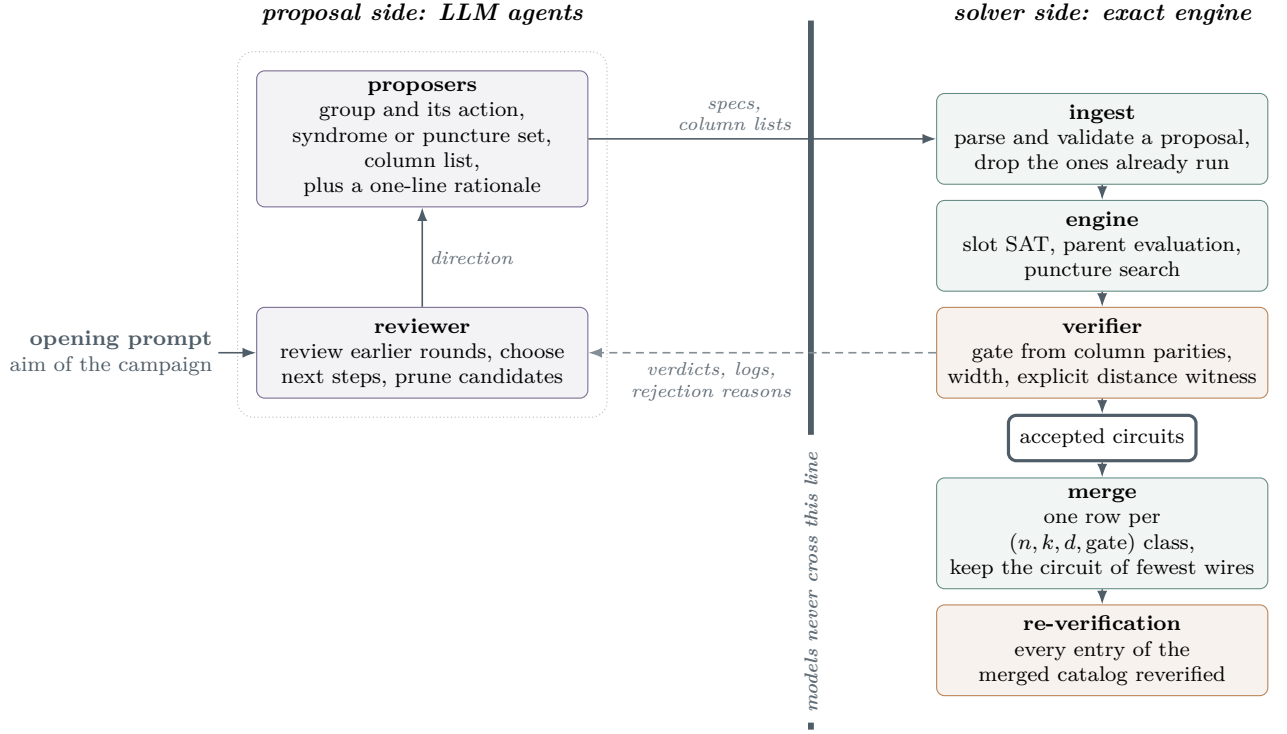
\begin{figure*}[t]
\centering
\begin{tikzpicture}[
  node distance=0.2cm and 0.9cm,
  bx/.style={draw=FlowSlate!70,rounded corners=3pt,fill=FlowMist,
             align=center,font=\footnotesize,inner sep=4pt,
             text width=4.1cm,minimum height=0.6cm},
  model/.style={bx,draw=FlowPlum,fill=FlowPlum!9},
  exact/.style={bx,draw=FlowSage,fill=FlowSage!9},
  verify/.style={bx,draw=FlowOchre,fill=FlowOchre!10},
  door/.style={bx,draw=FlowSlate,fill=white,very thick,text width=2.2cm},
  ar/.style={-{Latex[length=1.8mm]},semithick,FlowSlate},
  back/.style={-{Latex[length=1.8mm]},semithick,FlowSlate!70,densely dashed},
  tag/.style={font=\scriptsize\itshape,color=FlowSlate!90,inner sep=1.5pt,
              fill=white,align=center}]
\node[model] (prop) {\textbf{proposers}\\group and its action,\\syndrome or puncture set,\\ column list,\\plus a one-line rationale};
\node[tag,font=\small\itshape,text=black,above=0.55cm of prop] (ptag) {\textbf{proposal side: LLM agents}};
\node[exact,right=4.6cm of prop] (ing) {\textbf{ingest}\\parse and validate a proposal,\\drop the ones already run};
\node[exact,below=of ing] (eng) {\textbf{engine}\\slot SAT, parent evaluation,\\ puncture search};
\node[verify,below=of eng] (v1) {\textbf{verifier}\\gate from column parities,\\width, explicit distance witness};
\node[door,below=of v1] (fd) {accepted circuits};
\node[exact,below=of fd] (mg) {\textbf{merge}\\{one row per $(n,k,d,\text{gate})$ class},\\keep the circuit of fewest wires};
\node[verify,below=of mg] (v3) {\textbf{re-verification}\\every entry of the merged catalog reverified};
\node[tag,font=\small\itshape,text=black,anchor=center] at (ptag -| ing) {\textbf{solver side: exact engine}};
\node[model] (dir) at (prop |- v1) {\textbf{reviewer}\\review earlier rounds, choose next steps, prune candidates};
\node[align=right,font=\footnotesize,color=FlowSlate,
      inner sep=2pt,left=0.5cm of dir] (seed)
  {\textbf{opening prompt}\\aim of the campaign};
\draw[ar] (ing)--(eng); \draw[ar] (eng)--(v1); \draw[ar] (v1)--(fd);
\draw[ar] (fd)--(mg); \draw[ar] (mg)--(v3);
\draw[ar] (seed) -- (dir);
\draw[ar] (dir) -- (prop) node[midway,tag,right=1pt] {direction};
\coordinate (top) at ($(prop.north east)!0.64!(ing.north west)+(0,0.85)$);
\coordinate (bot) at ($(top |- v3.south)+(0,-0.35)$);
\coordinate (lab) at ($(top)!0.80!(bot)$);
\draw[line width=2.2pt,FlowSlate] (top) -- ($(lab)+(0,1.95)$);
\draw[line width=2.2pt,FlowSlate] ($(lab)+(0,-1.95)$) -- (bot);
\node[font=\scriptsize\itshape,color=FlowSlate!90,inner sep=0pt,rotate=90,anchor=center] at (lab)
  {\textbf{models never cross this line}};
\draw[ar] (prop.east) -- (ing.west);
\node[tag,anchor=south east,align=center] at ($(top |- prop.east)+(-0.2,0.06)$)
  {specs,\\ column lists};
\draw[back] (v1.west) -- (dir.east);
\node[tag,anchor=north east,align=center] at ($(top |- v1.west)+(-0.2,-0.07)$)
  {verdicts, logs,\\rejection reasons};
\begin{scope}[on background layer]
\node[fit=(dir)(prop),rounded corners=6pt,draw=FlowSlate!45,densely dotted,
      inner sep=7pt] {};
\end{scope}
\end{tikzpicture}
\caption{The agent loop.  An opening prompt, written by the user, sets the
aim of the campaign.  On the proposal side the reviewer reads the results
of earlier rounds and sets a direction, which the proposers use to 
generate proposals.  A proposal names a symmetry group and its
action, a syndrome or puncture set, or an explicit column list, and
carries a one-line rationale.  The heavy line separates the two sides,
and only files cross it.  Proposals go right as specification files and
column lists, and verdicts, logs and rejection reasons come back left.
No model runs on the solver side.  Ingest parses each proposal, rejects
one it cannot read and drops one already run.  The engine runs the tool
the proposal names.  Many proposals return no circuit, and those negatives
are passed back to the proposal side along with the successes.  A circuit that does come back goes to the verifier,
which re-derives its gate, width and distance from the matrix alone, with no knowledge of the agents' proposal.
Merge keeps one row per class, a class being $(n,k,d,\textrm{output gate})$
with the gate taken up to $\mathrm{GL}(k,2)$ and diagonal Clifford corrections, as in Sec.~\ref{sec:ai-landscape}.  Where
several circuits reach the same class, the one on fewest wires is kept.
The merged catalog is reverified once more at the end of the round.  One
round is one pass around the loop, and a campaign is a sequence of
rounds.}
\label{fig:loop}
\end{figure*}

 Orbits are disjoint and each carries
at most one label, so the selected columns have distinct nonzero syndromes
and thus, by Theorem~\ref{thm:syndrome-criterion}, distance $3$ is built-in.  The ansatz is therefore a
SAT problem.  One Boolean variable per candidate slot records whether that
slot is selected.  There are at most $2^kq_\cG$ variables instead of
$2^k(2^r-1)$, one at-most-one constraint per orbit, and parity constraints on
orbit sums for each index $Q$ of the target $D$.
We pass the instance to solvers such as CP-SAT or CryptoMiniSat.  A solver
may find a solution, establish infeasibility, or stop without resolving the
instance.  For trivial $\cG$ every orbit is a singleton and the instance is
the unrestricted one over all reduced distance-$\ge3$ factories with $r$
checks.  A nontrivial $\cG$ compresses the search but narrows its scope. 

The group of Fig.~\ref{fig:slots}, an order-three permutation that rotates
two triples of check wires and fixes the rest, shows the size of the
compression.  At $r=7$ it has $q_\cG=47$ orbits, seven fixed points and forty
of size three, so at $k=4$ the SAT instance carries $2^4\times47=752$ slot
variables where the unrestricted instance carries $2^4\times127=2032$.  The
parent count shrinks as well.  A parent is now a choice of orbits rather than
of syndromes, giving $2^{47}$ candidates instead of $2^{127}$, a reduction by
a factor of $2^{80}$.  On six check wires the same group leaves
$23$ orbits, $368$ slot variables against $1008$, and $2^{23}$ parents
against $2^{63}$. Among the factories of Sec.~\ref{sec:ai-landscape}, the slot search with this
symmetry found $\factory{28}{2}{3}$, $\factory{44}{4}{3}$, $\factory{51}{5}{3}$
and $\factory{85}{1}{5}$ for pure $\Tgate^{\otimes k}$, $\factory{47}{4}{3}$, and
the $\factory{63}{4}{3}$ $\CS^{\otimes2}$ factory.



\section{Agent-driven search}\label{sec:ai-methods}

Choosing groups by hand gives factories beyond the classified range,
but little guidance on which groups to try next.  We therefore use LLM agents
to select groups, parents and puncture sets based on earlier search outcomes.
Each search is organized as a \emph{campaign}, a sequence of \emph{rounds}.
In each round, \emph{proposer} agents write proposals from a fixed
\emph{brief}, a deterministic \emph{engine} runs the solver each proposal
names, an independent \emph{verifier} checks every returned circuit, and a
\emph{reviewer} agent reads the outcomes and writes the \emph{direction},
a short instruction for the next round.

A proposal names one of three kinds of object.  The first is a \emph{group}
with its action on the check wires and a column count, which the slot
solver of Sec.~\ref{sec:slots} converts into an orbit search.  The second is
a parent, a syndrome set whose compatible outputs are enumerated by the
parent evaluation of Sec.~\ref{sec:parent}.  The third is a \emph{puncture
set}.  Puncturing starts from one fixed code on many coordinates and deletes
some of them.  A deleted coordinate stops being an input and becomes an
output, so deleting $j$ independent coordinates from a code of length $\ell$
leaves a factory with $\ell-j$ inputs and $j$ additional
outputs~\cite{haah2018codes}.  This is how the largest factories in the
catalog were found, where our searches cannot enumerate parents and
labelings.  It is also where the lowest overhead exponents of
Table~\ref{tab:rates} come from, obtained by instantiating the frameworks of
Refs.~\cite{wills2026improved,san-jose2026asymptotically}.
Appendix~\ref{app:ai-puncture} describes these searches in detail, and
Appendix~\ref{app:proofs} gives the pruning rule used in the puncture search
and two extension moves that build new factories from verified ones.  A
proposal may also be an explicit circuit, which bypasses the engine and goes
straight to the verifier.

Each campaign consists of successive rounds of proposals and solver runs
(Fig.~\ref{fig:loop}).  The proposers work from a fixed brief describing the
factory representation, accepted proposal formats and search objective.  Each
round adds a direction written by the reviewer from the previous round's
results.  Both proposers and reviewers contribute to the proposal batch.  An
\emph{ingest} step removes malformed and duplicate proposals, and the engine
processes the remaining candidates.  An independent
verifier checks each returned circuit before it enters the catalog.  The
catalog is then reverified as a whole.  The reviewer reads the solver outcomes
and verification logs to set the next direction.  A campaign ends after its
assigned number of rounds or when stopped by the authors.
Appendix~\ref{app:round} follows one logged round from its direction to the
next.

The verifier re-derives the target gate and check parities from the returned
columns.  It proves the claimed distance bound and checks any witness used to
claim an exact distance.  It also checks that the output rows are independent
modulo the check span and that, in the reported output basis, every declared
output occurs in the recovered non-Clifford gate.  These checks exclude
dependent logical labels and explicit spectator outputs.

Unsuccessful searches also guide the reviewer's next steps.  A solver result
of \texttt{UNSAT} rules out a solution within the specified constraints,
whereas a timeout leaves the instance unresolved and may prompt a change in
strategy or computational budget.  In the puncture campaign, for example,
repeated randomized restarts timed out without a conclusive result.  The
reviewer responded by focusing the search on a smaller neighborhood, where
the solver completed the search and found an improved factory.



Across campaigns, we varied the number of reviewers and proposers, the models
used, and the number of rounds.  The full catalog of results is available in the
public database~\cite{jain2026shubhamj810}.

\section{Outlook}\label{sec:outlook}

This work brings together a unified binary-matrix formulation, symmetry-based
searches and LLM agents to discover magic-state factories beyond the reach of
exhaustive classification.  That classification now reaches $n\le54$ at
distance three and above, at a cost of order $10^5$ core-hours, and the same
method is estimated to need $10^8$ core-hours for the next input
count~\cite{wills2026classification}.  We extend the search beyond this limit by
combining symmetry reductions with LLM agents that use previous search
outcomes to propose promising groups, parents and puncture sets.  These
targeted searches yield both pure-$\Tgate$
factories and entangled outputs built from combinations of $\Tgate$, $\CS$
and $\CCZ$ gates.  The resulting catalog is a starting point.  Our
methods can readily be used to discover more factories by varying the search
goals and constraints.

An immediate next step is to translate the improvements in input count and
overhead exponent into savings in a fault-tolerant architecture.  This
requires accounting for acceptance probabilities, leading error coefficients
and the Clifford circuits that consume the outputs.  For entangled outputs,
the advantage of distilling a desired gate directly also depends on the gate
mix of the compiled algorithm.  The explicit circuits in our catalog provide
a basis for these comparisons and for searches guided by the cost of the
complete computation.

The framework also offers several directions for further exploration.  Within
the present class of postselected diagonal protocols, searches at distance
four and above can incorporate additional constraints on the output labels
(Sec.~\ref{sec:matrix-problem}).  Extending the output gates to level four,
including $\CCCZ$, requires $\pi/16$ rotations and phase polynomials of
degree four~\cite{cui2017diagonala}.  Broadening the search to catalytic and
Hadamard-assisted protocols could yield further
improvements~\cite{singh2026borrowed,gidney2019efficient}, with verifiers
adapted to check these circuits directly.  Across these extensions, the same
workflow can be retained: agents guide the search, deterministic solvers
construct candidates, and independent verifiers establish their properties.

To support continued discovery, we release the catalog, explicit circuits,
search engine and verifiers together with a \emph{context directory} for the
quantum community~\cite{jain2026shubhamj810}.  The directory brings together the
mathematical framework, search techniques and lessons from previous campaigns,
including established results and open directions.  These resources enable others to brief their own agents,
build on our searches and adapt them to their gate mixes and resource
requirements.  We invite the community to extend these methods and contribute
verified circuits, building an active, open-source repository that grows with
each new effort.

\section*{Use of AI}

The framework, the search designs, the solver, the verifiers and the
acceptance criteria in this work are the authors' own. Within that framework, large language models were the
proposing half of the loop of Sec.~\ref{sec:ai-methods}: they proposed parents,
symmetry groups and search directions and, in the later runs, pruned
candidates as reviewers. They also assisted with the scripts around the loop, the catalog build
and the figure tooling. Every circuit here was re-derived
from its stored columns by the verifiers of Sec.~\ref{sec:ai-methods} before
entering the catalog.  The authors are responsible for the results and claims in this paper.

The agents were Anthropic Claude models throughout.  Reviewers were
typically the latest available Fable or Opus versions, and proposers were
Opus or Sonnet models.  Some campaigns were run with
\textsc{humanize2}~\cite{zhang2026humanize}, an open-source orchestrator for
workflows of coding agents, in a loop of the same shape as the one described
in Sec.~\ref{sec:ai-methods}.

\section*{Data availability}
The catalog, the engine, the verifiers, the context directory of briefs and
campaign notes, and every circuit in this paper are available in the public
database of Ref.~\cite{jain2026shubhamj810}.

\begin{acknowledgments}
We thank Anqi Gong, Andrew Cross, Juan-Cruz Benito, David Kremer and Swanand
Ravindra Kadhe for helpful discussions. A.W. acknowledges support from the MIT-IBM Watson AI Lab. This preprint is assigned number MIT-CTP/6112. The theoretical framework, the symmetric search design and the acceptance criterias are the authors' own. We acknowledge the use of Anthropic AI models for help in building the code and generating proposals in our agent loop. Full details of the AI usage can be found in our section titled "Use of AI".
\end{acknowledgments}

\bibliography{refs}

\appendix
\onecolumngrid
\raggedbottom
\section{Pruning and extension rules used in the search}\label{app:proofs}

We prove one pruning rule and two extension moves.  Downward closure gives the
pruning rule used in the $\RM(3,10)$ search of
Appendix~\ref{app:ai-puncture}.  Free-row appending was used after puncturing
the seven-check simplex parent, which contains all $127$ nonzero columns in
$\F_2^7$.  One resulting factory is $\factory{112}{16}{3}$.  Parity lift gives
another certified extension and maps $\factory{99}{12}{3}$ to
$\factory{100}{12}{4}$.

For the puncture search, write $g_x$ for the column of degree-$\le3$ monomial
evaluations at coordinate $x$, $g_X=\bigoplus_{x\in X}g_x$ for its sum over a
set $X$, and $U_J=\Span\{g_j:j\in J\}$ for the span over the deleted set $J$.
We assume that the columns indexed by $J$ are independent, as they are in the
puncture searches below, so each deleted coordinate contributes one output.
A retained fault support $X$ is accepted and harmful
(Sec.~\ref{sec:matrix-problem}) exactly when $0\ne g_X\in U_J$.  For a given support $X$, this requires a test of nonzero membership in $U_J$.

\begin{theorem}[downward closure]\label{thm:downward}
If a puncture set $J$ is enlarged to $J\cup L$ with the columns
$\{g_l\}_{l\in L}$ independent modulo $U_J=\Span\{g_j:j\in J\}$, the span
of the deleted columns, so that each new coordinate is a new output, then $d(J\cup L)\le d(J)$, where $d(J)$ is the distance of
the factory punctured at $J$.  Hence distance-preserving
punctures are downward closed.
\end{theorem}

\begin{proof}
Let $X$ be a minimum accepted harmful support for $J$, so $0\ne g_X\in U_J$,
and put $R=X\cap L$.  If $R=\varnothing$, $X$ remains accepted and harmful.
If $0<|R|<|X|$ the
residual $X\setminus R$ has column sum $g_X+g_R$ in the enlarged span, nonzero
because independence makes every nonempty sum of new columns nonzero modulo
$U_J$, so it is accepted, harmful and shorter.  If $R=X$ then $g_X\in U_J$
contradicts
that independence.
\end{proof}

\begin{proposition}[parity lift]\label{prop:parity}
Let $G$ be a reduced, full-rank $\factory nkd$ circuit with check span $S$,
output rows $g_1,\dots,g_k$ and odd distance $d$.  Suppose its target is
\emph{homogeneous}: every degree-one and degree-two output moment
vanishes, so the output is a product of $\CCZ$ gates on triples that need
not be disjoint.  If $n$ is even and
$\one\notin\Span(S,g_1,\dots,g_k)$,
appending the all-ones row $\one$ as a check gives a factory for the same
target with $n$ inputs, one more check wire and distance at least
$d+1$.  If $n$ is odd, first append a column supported only on the new
check wire.  The result has $n+1$ inputs, one more check wire and
distance at least $d+1$.
\end{proposition}

\begin{proposition}[free-row appending]\label{prop:freerow}
Let $G$ be a reduced, full-rank distance-$\ge3$ factory with check span $S$ and
output rows $g_1,\dots,g_k$, and prescribe the moments a new output is to
have.  These are a degree-one value $t$, degree-two values $v_i$ against each $g_i$
and degree-three values $w_{ij}$ against each pair $g_i,g_j$ with $i<j$, all in
$\F_2$.  The rows $a\in\F_2^n$ with
$\wt(a)\equiv t$, $\wt(a\od g_i)\equiv v_i$,
$\wt(a\od g_i\od g_j)\equiv w_{ij}$ and
$\wt(a\od s)\equiv\wt(a\od s\od s')\equiv\wt(a\od s\od g_i)\equiv0
\pmod2$ for all $s,s'\in S$ form an affine subspace of $\F_2^n$, possibly
empty, and those outside the span of the existing check and output rows
extend $G$ to a reduced, full-rank distance-$\ge3$ circuit with $k+1$ output
rows and the prescribed moments.  For any row $h$, define its moment vector
relative to the old outputs by
\[
 M(h)=\bigl(\wt(h),(\wt(h\od g_i))_i,
 (\wt(h\od g_i\od g_j))_{i<j}\bigr)\pmod 2.
\]
Suppose no nonzero combination of the old output rows has zero moment vector.
If
\[
 (t,(v_i)_i,(w_{ij})_{i<j})
 \ne M\!\left(\sum_{i\in I}g_i\right)
 \quad\text{for every }I\subseteq\{1,\dots,k\},
\]
including $I=\varnothing$, then no nonzero combination of the extended output
rows is idle.  The extended circuit is therefore a
$\factory n{k+1}{\ge3}$ factory with no idle output direction.
\end{proposition}

\begin{proof}[Proof of Proposition~\ref{prop:parity}]
A row $c$ may be appended as a check exactly when $\wt(c)$, $\wt(c\od g)$
and $\wt(c\od g\od g')$ are even for all existing rows $g,g'$, since a
check carries no non-Clifford phase.  For $c=\one$ these are its length,
$\wt(g)$ and $\wt(g\od g')$, which for output rows are the degree-one and
degree-two moments and vanish by homogeneity, and for checks are even already, since no monomial
touching a check wire carries a non-Clifford phase.  Thus $\one$ is a legal
check when its length is even.  In the even-$n$ case the added row is
independent of all existing rows by hypothesis, while appending the same bit
to every column preserves distinctness and nonzeroness.

In the odd-$n$ case, extend every old row by a zero at the appended column and
take $c=\one$ of length $n+1$.  The new column is distinct because the old
circuit has no zero column, and $c$ is independent of all old rows because it
is the only row nonzero at the new coordinate.  The new column enters no
overlap of the old rows and has zero output label, so the target is unchanged.

In either case every column has a one on the new check wire, and hence every
accepted fault set has even size.  Such a set is accepted by the old checks as
well.  In the even-$n$ case it is therefore an accepted harmful set of $G$,
so it cannot have size below $d$ or size $d$, since $d$ is odd.  In the odd-$n$ case, let
$e$ denote the appended column.  If an accepted harmful set $F$ excludes $e$,
then it is an accepted harmful set of $G$ of even size and therefore has
$|F|\ge d+1$.  If $e\in F$, then $F\setminus\{e\}$ is an accepted harmful set
of $G$, so $|F|-1\ge d$.  Thus the new distance is at least $d+1$ in both
cases.
\end{proof}

When $d(G)=d$ exactly, $\one$ is automatically independent of the check rows,
since otherwise every accepted set would already have even size.  Only its
independence from the output rows needs to be checked.

\begin{proof}[Proof of Proposition~\ref{prop:freerow}]
Each condition is the parity of the coordinatewise product of $a$ with a
fixed vector, hence an affine-linear equation over $\F_2$, and a finite
system of such equations has an affine solution set.  Appending an output
row changes no syndrome, so Theorem~\ref{thm:syndrome-criterion} keeps the
distance at least three, and the columns stay distinct and nonzero.  A row
outside the span of the existing rows increases the row rank by one and hence
adds one logical qubit.  A row inside that span does not.
The moment map $M$ is linear over $\F_2$.  Every output direction involving
the appended row has the form $a+\sum_{i\in I}g_i$, whose moment vector is
\[
 (t,(v_i)_i,(w_{ij})_{i<j})
 +M\!\left(\sum_{i\in I}g_i\right).
\]
The stated condition makes each such vector nonzero.  Output directions not
involving $a$ are nonidle by hypothesis.  To see that $M$ also detects
idleness in the full extended space, write a direction involving the new row
as $h=a+s$, with $s$ in the old output span.  Since $h\od h=h$, for every old
output-span vector $u$ we have
\[
 \wt(h\od a)\equiv\wt(h)+\wt(h\od s),\qquad
 \wt(h\od a\od u)\equiv
 \wt(h\od u)+\wt(h\od s\od u)\pmod2.
\]
The terms on the right are linear combinations of components of $M(h)$.
Multilinearity then covers every extended output-span vector.  Hence no
nonzero direction is idle in any output basis.
\end{proof}

\section{Puncture searches}\label{app:ai-puncture}

For a parent with hundreds of coordinates, full enumeration of output
labelings is beyond our search capacity.  We instead fix a base code and
search over coordinates to delete, so that each independent puncture turns
one input into one output.  The first search uses $\RM(3,10)$ on $1024$
coordinates, one of the two Reed--Muller parents used in
Ref.~\cite{haah2018codes} and the base of its large-block table.  We keep a
puncture set only if the verifier certifies the required lower bound on the
distance.

Our search differs from Ref.~\cite{haah2018codes} in how the puncture set is
varied.  Haah and Hastings searched random punctures of the smaller
$\RM(2,7)$ parent.  For $\RM(3,10)$, they began with a randomly punctured
distance-three code and successively unpunctured coordinates to raise the
distance to four and five.  Their distance-six and distance-seven codes were
obtained by further unpuncturing their best-rate distance-five code.  We run in
the opposite direction, deleting further from a puncture set already in hand
while holding the distance fixed.  Downward closure
(Theorem~\ref{thm:downward}) prunes these extensions.  A coordinate whose
deletion alone breaks the bound cannot appear in any admissible extension,
and a pair that breaks it cannot appear together, so single and pairwise
tests reduce the search to cliques of a compatibility graph.  Starting from
the Haah--Hastings puncture set $\factory{887}{137}{5}$, the scan identifies
$37$ safe coordinates among the $887$ unpunctured coordinates.  They form a
$37$-vertex, $469$-edge compatibility graph.  An exact audit of all $5{,}314$
cliques of size at least seven finds $25$ admissible sets of size seven and
none larger.  This certifies $k\le144$ within this descendant family, with the
bound attained by $\factory{880}{144}{5}$ at $\gamma=1.125$.  A separate search
found the distance-five entry of Table~\ref{tab:rates},
$\factory{879}{145}{5}$ at $\gamma=1.120$.  It is another puncture set of the
same base code, but not a descendant of the $\factory{887}{137}{5}$ seed.
Puncturing fewer coordinates can raise the distance instead.  With $122$
punctures, the base code yields $\factory{902}{122}{6}$.  Relative to
$\factory{880}{144}{5}$, it uses $22$ more inputs and produces $22$ fewer
outputs, with $\gamma=1.117$.  A
\emph{ruin-and-recreate move} deletes part of a puncture set and rebuilds
it.  Restoring four of the $122$ coordinates and deleting five others gives
$\factory{901}{123}{6}$ with $\gamma=1.111$.  Its leading harmful-fault
count is larger, with $2781$ accepted harmful fault sets of weight six against
$2475$.

Some of the lowest exponents in the catalog come from puncturing other
parents.  These include the weighted Reed--Muller codes of
Ref.~\cite{wills2026improved} and the derivative four-block Schur-orthogonal extensions.  Lemma~2.3 of
Ref.~\cite{san-jose2026asymptotically} gives an independent derivation of the
required Clifford phase corrections.  The downset construction of Proposition~3.7 in
Ref.~\cite{wills2026improved} restricts the puncture set itself.  If it contains
a Boolean vector, it also contains every coordinatewise smaller vector, which
yields a closed-form distance formula.  Our graph search allows arbitrary
independent puncture sets, so we certify each distance directly from the
resulting columns.  Two further moves are used.  A
\emph{logical-pair check} deletes two output rows and appends their sum as a
check.  A \emph{paired-column contraction} restricts the row space so that two
chosen columns agree, then deletes those columns and a check row that separates
them.  Successful contractions save two inputs while preserving the verified
output count and distance.  The distance-six entry of Table~\ref{tab:rates}
uses both, and the largest factory of this kind in our catalog,
$\factory{1715}{287}{6}$, reaches $\gamma=0.998$, the smallest explicit
pure-$\Tgate$ factory with an exponent below one that we know of.

\section{A worked round}\label{app:round}

Figure~\ref{fig:round} follows the second round of one slot-search campaign that
improved the best known input counts in its search cells.  The top box is the brief of
Sec.~\ref{sec:ai-methods}, fixed for the whole campaign, and everything below
it is that round's state.  Two words in the figure are the campaign's own.  A
\emph{geometry} is a slot specification, naming an output count $k$, an
ordered list of check blocks and a permutation group acting on each block.  A
\emph{cell} is one combination of output count, output gate and distance, for
which the campaign keeps a best known input count.

Eighty-four geometries ran in this round.  The summary line counts geometries and the
direction below it counts cells, which is why the two read differently.  Six
geometries beat the standing record, in two cells, four of them in the $\CCZ$
cell, and four more opened one cell that had been empty.  Three of the classes
this paper reports come from the round.  $\factory{64}{3}{4}$ for $\CCZ$
improves on the catalog's earlier $\factory{66}{3}{4}$, and the round also returned
$\factory{60}{2}{4}$ for $\Tgate^{\otimes2}$ and $\factory{48}{2}{4}$ for
$\CS$.  All three were found under a cap rather than proved minimal, and the
review recorded that at the time.

The rest of the round consisted of unsuccessful solver outcomes.  Forty-five
geometries returned \texttt{UNSAT} in seconds, ruling out solutions within
the specified constraints (Sec.~\ref{sec:ai-methods}).  The review compared
this distance-four sweep with the previous round's hundred-plus distance-three
runs, all of which had returned \texttt{UNKNOWN} at the time limit.  On this evidence, the reviewer restricted the next round
to distance four and above.

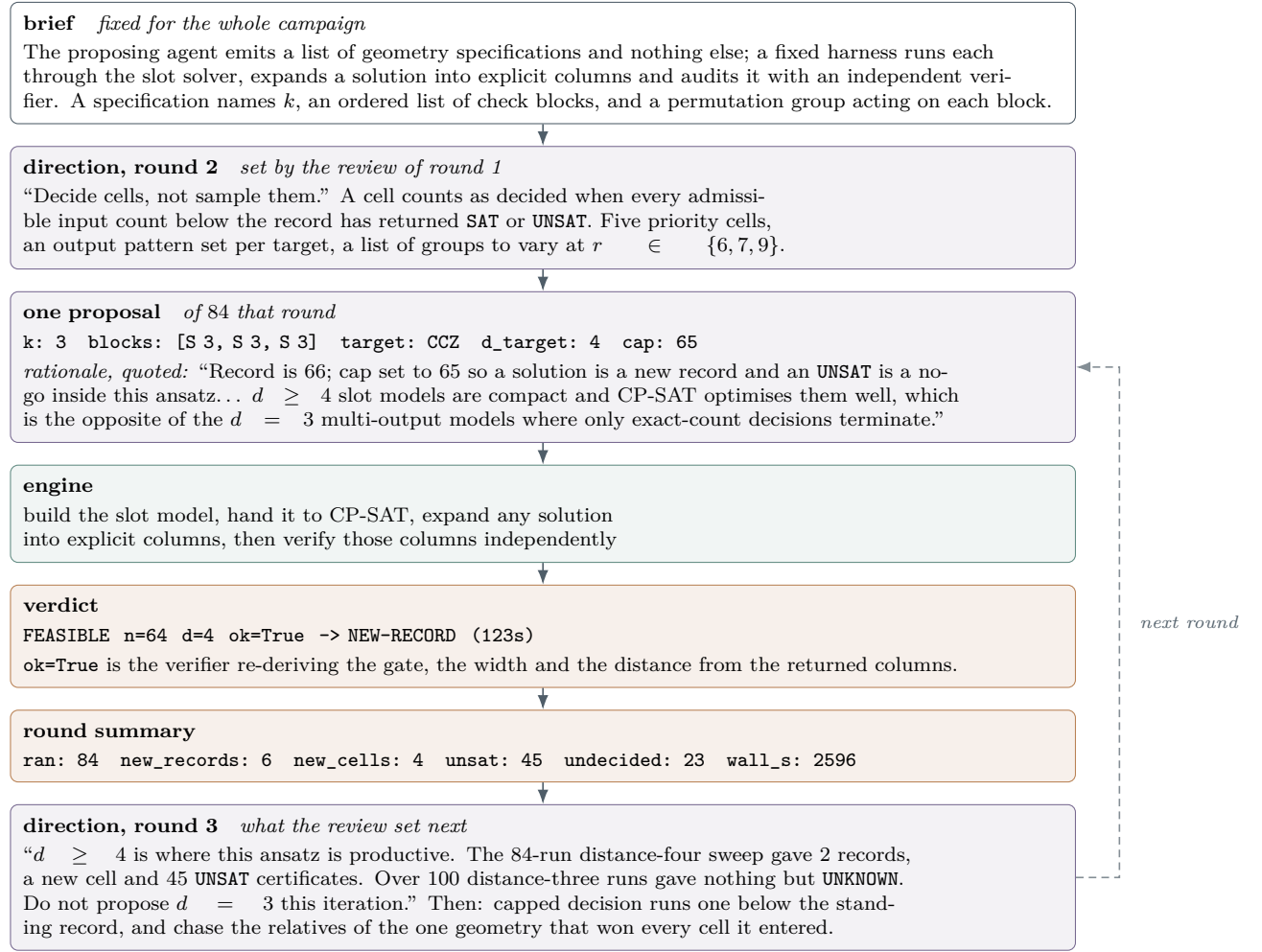
\begin{figure}[h]
\centering
\begin{tikzpicture}[
  font=\footnotesize,
  bx/.style={draw=FlowSlate!70,rounded corners=3pt,fill=FlowMist,align=left,
             inner sep=5pt,text width=0.80\linewidth},
  fixed/.style={bx,draw=FlowSlate,fill=white},
  model/.style={bx,draw=FlowPlum,fill=FlowPlum!9},
  exact/.style={bx,draw=FlowSage,fill=FlowSage!9},
  verify/.style={bx,draw=FlowOchre,fill=FlowOchre!12},
  ar/.style={-{Latex[length=2mm]},semithick,FlowSlate},
  back/.style={-{Latex[length=2mm]},semithick,FlowSlate!70,densely dashed},
  lab/.style={font=\scriptsize\itshape,color=FlowSlate,inner sep=2pt}]

\node[fixed] (brief)
  {\textbf{brief}\quad\emph{fixed for the whole campaign}\\[2pt]
   The proposing agent emits a list of geometry specifications and nothing
   else; a fixed harness runs each through the slot solver, expands a
   solution into explicit columns and audits it with an independent
   verifier.  A specification names $k$, an ordered list of check blocks,
   and a permutation group acting on each block.};

\node[model,below=3mm of brief] (dir)
  {\textbf{direction, round 2}\quad\emph{set by the review of round 1}\\[2pt]
   ``Decide cells, not sample them.''  A cell counts as decided when every
   admissible input count below the record has returned \texttt{SAT} or
   \texttt{UNSAT}.  Five priority cells, an output pattern set per target, a
   list of groups to vary at $r\in\{6,7,9\}$.};

\node[model,below=3mm of dir] (prop)
  {\textbf{one proposal}\quad\emph{of $84$ that round}\\[2pt]
   \texttt{k: 3}\quad \texttt{blocks: [S 3, S 3, S 3]}\quad
   \texttt{target: CCZ}\quad \texttt{d\_target: 4}\quad \texttt{cap: 65}\\[2pt]
   \emph{rationale, quoted:} ``Record is $66$; cap set to $65$ so a solution
   is a new record and an \texttt{UNSAT} is a no-go inside this ansatz\ldots\
   $d\ge4$ slot models are compact and CP-SAT optimises them well, which is
   the opposite of the $d=3$ multi-output models where only exact-count
   decisions terminate.''};

\node[exact,below=3mm of prop] (eng)
  {\textbf{engine}\quad\emph{}\\[2pt]
   build the slot model, hand it to CP-SAT, expand any solution into explicit
   columns, then verify those columns independently};

\node[verify,below=3mm of eng] (ver)
  {\textbf{verdict}\\[2pt]
   \texttt{FEASIBLE\ \ n=64\ \ d=4\ \ ok=True\ \ ->\ NEW-RECORD\ \ (123s)}\\[2pt]
   \texttt{ok=True} is the verifier re-deriving the gate, the width and the
   distance from the returned columns.};

\node[verify,below=3mm of ver] (sum)
  {\textbf{round summary}\\[2pt]
   \texttt{ran: 84}\quad \texttt{new\_records: 6}\quad \texttt{new\_cells: 4}\quad
   \texttt{unsat: 45}\quad \texttt{undecided: 23}\quad \texttt{wall\_s: 2596}};

\node[model,below=3mm of sum] (nxt)
  {\textbf{direction, round 3}\quad\emph{what the review set next}\\[2pt]
   ``$d\ge4$ is where this ansatz is productive.  The $84$-run distance-four
   sweep gave $2$ records, a new cell and $45$ \texttt{UNSAT} certificates.
   Over $100$ distance-three runs gave nothing but \texttt{UNKNOWN}.  Do not
   propose $d=3$ this iteration.''  Then: capped decision runs one below
   the standing record, and chase the relatives of the one geometry that won
   every cell it entered.};

\draw[ar] (brief)--(dir);
\draw[ar] (dir)--(prop);
\draw[ar] (prop)--(eng);
\draw[ar] (eng)--(ver);
\draw[ar] (ver)--(sum);
\draw[ar] (sum)--(nxt);
\draw[back] (nxt.east) -- ++(0.6,0) |- (prop.east)
  node[pos=0.25,lab,anchor=west,xshift=2mm] {next round};
\end{tikzpicture}
\caption{One logged round of the agent loop, quoted from the campaign's
brief, its direction files, its proposal set and its run log.  The brief is
fixed for the whole campaign.  Everything below it belongs to this round.  The
proposal shown is one of $84$, and it returned a factory at $n=64$ where the
catalog held $66$.  Most of the round's runs failed cheaply,
forty-five \texttt{UNSAT} in seconds.  That is what let the review close the
next round to distance four and above.  The same round produced the
$\factory{60}{2}{4}$ and $\factory{48}{2}{4}$ records.}
\label{fig:round}
\end{figure}

\section{Full catalog}\label{app:catalogue}

Tables~\ref{tab:catalogue} and~\ref{tab:catalogue-mixed} print the
level-three, distance-$\ge3$ factories our searches produced,
ordered by input count.  There are $699$ classes, a class being defined as
in Sec.~\ref{sec:ai-landscape}.  A class
reached by several searches appears once, represented by the circuit of fewest
wires.

Of the $699$ classes, $414$ have pure-$\Tgate^{\otimes k}$ outputs.
Table~\ref{tab:catalogue} lists $220$ representatives, giving their parameters
$\factory nkd$ and width $N$.  Their phase polynomial
is $x_0+\dots+x_{k-1}$ and their $\Tgate$-count is $k$.  Only the highest
$k$ at each $n$ and $d$ is printed, since discarding outputs trivially gives smaller
$k$. Table~\ref{tab:catalogue-mixed} lists
the other $285$ classes,
including their output phase polynomials except where summarized below.

The \emph{output phase polynomial} specifies the delivered gate up to
Clifford corrections.  It contains the same data as the target vector $D$ of
Sec.~\ref{sec:matrix}.  Its degree-one, degree-two and degree-three terms are
respectively the $\Tgate$, $\CS$ and $\CCZ$ content of the output gate.  We write it in one notation throughout.  The wires of a factor are separated
by spaces, and \texttt{+} joins the factors.  The number of wires in a factor
gives its type, so one wire is a $\Tgate$, two a $\CS$ and three a $\CCZ$.  Thus \texttt{0+1+0 1} is
$\Tgate_0\Tgate_1\CS_{01}$, \texttt{0 1 2} is $\CCZ_{012}$, and
\texttt{0 64 83} is $\CCZ_{0,64,83}$.  A polynomial containing only one-wire factors is a product of $\Tgate$ gates.
Two- and three-wire factors specify $\CS$ and $\CCZ$ gates.  The wires are written out rather than run together because
$k$ reaches $84$ in Table~\ref{tab:catalogue-mixed}, where \texttt{10} would be unreadable as either wire
ten or the pair $\{1,0\}$.

The remaining columns are the width $N=k+r$ and the
$\Tgate$-count of the output gate~\cite{amy2019tcount}.  A dash
marks an output wider than $k=6$, beyond the range of our minimization.  Nine output polynomials have dozens to thousands of factors,
the largest $3587$. Those rows record the number of
factors instead, and the gate itself is in the database.  A distance
written $\ge d$ is a proved lower bound rather than a measured value. The tables include rediscoveries alongside new results.  The novelty
criterion and class counts are given in Sec.~\ref{sec:ai-landscape}.

\par\medskip
\begingroup\centering\footnotesize
\setlength{\tabcolsep}{6pt}
\renewcommand{\arraystretch}{1.0}
\refstepcounter{table}\label{tab:catalogue}%
\begin{tabular}{@{} r l r @{\quad}|@{\quad} r l r @{\quad}|@{\quad} r l r@{}}
\multicolumn{9}{c}{TABLE~\thetable. Factories with a pure $\Tgate^{\otimes k}$ output: $k$ independent}\\
\multicolumn{9}{c}{magic states.  Only the highest $k$ at each $n$ and $d$ is listed.}\\[3pt]
\toprule
\# & factory & $N$ & \# & factory & $N$ & \# & factory & $N$\\
\midrule
1 & $\factory{15}{1}{3}$ & 5 & 47 & $\factory{66}{8}{3}$ & 16 & 93 & $\factory{127}{1}{7}$ & 22\\
2 & $\factory{23}{1}{3}$ & 7 & 48 & $\factory{67}{9}{3}$ & 16 & 94 & $\factory{127}{5}{5}$ & 19\\
3 & $\factory{27}{1}{3}$ & 7 & 49 & $\factory{68}{10}{3}$ & 17 & 95 & $\factory{127}{23}{3}$ & 30\\
4 & $\factory{28}{2}{3}$ & 9 & 50 & $\factory{69}{9}{3}$ & 17 & 96 & $\factory{135}{17}{3}$ & 25\\
5 & $\factory{29}{1}{3}$ & 9 & 51 & $\factory{70}{10}{3}$ & 18 & 97 & $\factory{141}{2}{4}$ & 12\\
\bottomrule
\end{tabular}
\par\medskip
\begin{tabular}{@{} r l r @{\quad}|@{\quad} r l r @{\quad}|@{\quad} r l r@{}}
\multicolumn{9}{c}{TABLE~\thetable\ (continued)}\\[3pt]
\toprule
\# & factory & $N$ & \# & factory & $N$ & \# & factory & $N$\\
\midrule
139 & $\factory{386}{94}{3}$ & 105 & 167 & $\factory{655}{177}{3}$ & 189 & 195 & $\factory{887}{137}{5}$ & 176\\
140 & $\factory{399}{77}{3}$ & 98 & 168 & $\factory{679}{185}{3}$ & 197 & 196 & $\factory{896}{128}{6}$ & 189\\
141 & $\factory{404}{100}{3}$ & 111 & 169 & $\factory{684}{132}{3}$ & 168 & 197 & $\factory{901}{123}{6}$ & 176\\
142 & $\factory{405}{101}{3}$ & 112 & 170 & $\factory{691}{189}{3}$ & 201 & 198 & $\factory{902}{122}{6}$ & 176\\
143 & $\factory{406}{102}{3}$ & 113 & 171 & $\factory{702}{194}{3}$ & 205 & 199 & $\factory{904}{120}{\geq 6}$ & 176\\
144 & $\factory{407}{73}{4}$ & 92 & 172 & $\factory{715}{197}{3}$ & 209 & 200 & $\factory{905}{119}{\geq 6}$ & 176\\
\bottomrule
\end{tabular}
\par\medskip
\refstepcounter{table}\label{tab:catalogue-mixed}%
\begin{tabular}{@{}r l r p{0.50\textwidth} r@{}}
\multicolumn{5}{c}{TABLE~\thetable. Factories with other outputs, deduplicated up to CNOT$+S$ output equivalence.}\\[3pt]
\toprule
\# & factory & $N$ & output phase polynomial & $\Tgate$-count\\
\midrule
1 & $\factory{31}{2}{3}$ & 7 & \texttt{0+\allowbreak{}1+\allowbreak{}0~1} & 1\\
2 & $\factory{31}{3}{3}$ & 8 & \texttt{0+\allowbreak{}1+\allowbreak{}2+\allowbreak{}0~1+\allowbreak{}0~2+\allowbreak{}1~2+\allowbreak{}0~1~2} & 1\\
3 & $\factory{31}{4}{3}$ & 9 & \texttt{0+\allowbreak{}1+\allowbreak{}2+\allowbreak{}3+\allowbreak{}0~1+\allowbreak{}0~2+\allowbreak{}0~3+\allowbreak{}1~2+\allowbreak{}1~3+\allowbreak{}2~3+\allowbreak{}0~1~2+\allowbreak{}0~1~3+\allowbreak{}0~2~3+\allowbreak{}1~2~3} & 1\\
4 & $\factory{31}{5}{3}$ & 10 & \texttt{0+\allowbreak{}1+\allowbreak{}2+\allowbreak{}3+\allowbreak{}4+\allowbreak{}0~1+\allowbreak{}0~2+\allowbreak{}0~3+\allowbreak{}0~4+\allowbreak{}1~2+\allowbreak{}1~3+\allowbreak{}1~4+\allowbreak{}2~3+\allowbreak{}2~4+\allowbreak{}3~4+\allowbreak{}0~1~2+\allowbreak{}0~1~3+\allowbreak{}0~1~4+\allowbreak{}0~2~3+\allowbreak{}0~2~4+\allowbreak{}0~3~4+\allowbreak{}1~2~3+\allowbreak{}1~2~4+\allowbreak{}1~3~4+\allowbreak{}2~3~4} & 1\\
5 & $\factory{35}{2}{3}$ & 8 & \texttt{0+\allowbreak{}1+\allowbreak{}0~1} & 1\\
6 & $\factory{35}{2}{3}$ & 8 & \texttt{0~1} & 3\\
\bottomrule
\end{tabular}
\par\medskip
\begin{tabular}{@{}r l r p{0.50\textwidth} r@{}}
\multicolumn{5}{c}{TABLE~\thetable\ (continued)}\\[3pt]
\toprule
\# & factory & $N$ & output phase polynomial & $\Tgate$-count\\
\midrule
43 & $\factory{55}{4}{3}$ & 11 & \texttt{0+\allowbreak{}1~2~3} & 7\\
44 & $\factory{55}{6}{3}$ & 13 & \texttt{0+\allowbreak{}1+\allowbreak{}2+\allowbreak{}3+\allowbreak{}4+\allowbreak{}0~1+\allowbreak{}0~2+\allowbreak{}1~2+\allowbreak{}1~4+\allowbreak{}2~3+\allowbreak{}3~4+\allowbreak{}0~1~2+\allowbreak{}0~1~3+\allowbreak{}0~2~4+\allowbreak{}0~3~5+\allowbreak{}0~4~5+\allowbreak{}1~3~5+\allowbreak{}2~4~5+\allowbreak{}3~4~5} & 9\\
45 & $\factory{55}{6}{3}$ & 13 & \texttt{0+\allowbreak{}1+\allowbreak{}2+\allowbreak{}3+\allowbreak{}4+\allowbreak{}5+\allowbreak{}0~1+\allowbreak{}0~2+\allowbreak{}0~3+\allowbreak{}0~4+\allowbreak{}0~5+\allowbreak{}1~2+\allowbreak{}1~3+\allowbreak{}1~4+\allowbreak{}1~5+\allowbreak{}2~3+\allowbreak{}2~4+\allowbreak{}2~5+\allowbreak{}0~1~2+\allowbreak{}3~4~5} & 9\\
46 & $\factory{55}{6}{3}$ & 13 & \texttt{0+\allowbreak{}1+\allowbreak{}2+\allowbreak{}3+\allowbreak{}4+\allowbreak{}5+\allowbreak{}0~1+\allowbreak{}0~2+\allowbreak{}0~3+\allowbreak{}0~4+\allowbreak{}0~5+\allowbreak{}1~2+\allowbreak{}1~3+\allowbreak{}1~4+\allowbreak{}1~5+\allowbreak{}2~3+\allowbreak{}2~4+\allowbreak{}2~5+\allowbreak{}3~4+\allowbreak{}3~5+\allowbreak{}0~1~2+\allowbreak{}0~1~3+\allowbreak{}0~1~4+\allowbreak{}0~1~5+\allowbreak{}0~2~3+\allowbreak{}0~2~4+\allowbreak{}0~2~5+\allowbreak{}0~3~4+\allowbreak{}0~3~5+\allowbreak{}1~2~3+\allowbreak{}1~2~4+\allowbreak{}1~3~5+\allowbreak{}2~4~5+\allowbreak{}3~4~5} & 11\\
\bottomrule
\end{tabular}
\par\medskip
\begin{tabular}{@{}r l r p{0.50\textwidth} r@{}}
\multicolumn{5}{c}{TABLE~\thetable\ (continued)}\\[3pt]
\toprule
\# & factory & $N$ & output phase polynomial & $\Tgate$-count\\
\midrule
85 & $\factory{60}{5}{3}$ & 12 & \texttt{0+\allowbreak{}1~2+\allowbreak{}3~4+\allowbreak{}0~1~2+\allowbreak{}0~1~3+\allowbreak{}0~2~4} & 10\\
86 & $\factory{60}{5}{3}$ & 12 & \texttt{0+\allowbreak{}1~2+\allowbreak{}3~4+\allowbreak{}0~1~2+\allowbreak{}0~1~3+\allowbreak{}0~2~4+\allowbreak{}0~3~4} & 9\\
87 & $\factory{60}{5}{3}$ & 12 & \texttt{0+\allowbreak{}1~3+\allowbreak{}0~1~2+\allowbreak{}0~3~4+\allowbreak{}1~2~3} & 10\\
88 & $\factory{60}{5}{3}$ & 12 & \texttt{0+\allowbreak{}1~3+\allowbreak{}0~1~2+\allowbreak{}2~3~4} & 10\\
89 & $\factory{60}{5}{3}$ & 12 & \texttt{0+\allowbreak{}1~3+\allowbreak{}2~4+\allowbreak{}0~1~2+\allowbreak{}0~3~4} & 10\\
90 & $\factory{60}{5}{3}$ & 12 & \texttt{0+\allowbreak{}1~3+\allowbreak{}2~4+\allowbreak{}0~1~2+\allowbreak{}1~2~3} & 9\\
\bottomrule
\end{tabular}
\par\medskip
\begin{tabular}{@{}r l r p{0.50\textwidth} r@{}}
\multicolumn{5}{c}{TABLE~\thetable\ (continued)}\\[3pt]
\toprule
\# & factory & $N$ & output phase polynomial & $\Tgate$-count\\
\midrule
127 & $\factory{63}{6}{3}$ & 16 & \texttt{0+\allowbreak{}1+\allowbreak{}2+\allowbreak{}3+\allowbreak{}4+\allowbreak{}5+\allowbreak{}0~1+\allowbreak{}0~2+\allowbreak{}0~4+\allowbreak{}1~2+\allowbreak{}1~3+\allowbreak{}1~4+\allowbreak{}1~5+\allowbreak{}2~3+\allowbreak{}2~4+\allowbreak{}2~5+\allowbreak{}3~4+\allowbreak{}4~5+\allowbreak{}0~1~2+\allowbreak{}0~1~3+\allowbreak{}0~2~3+\allowbreak{}0~3~4+\allowbreak{}0~3~5+\allowbreak{}1~2~4+\allowbreak{}1~3~4+\allowbreak{}2~4~5} & 12\\
128 & $\factory{63}{6}{3}$ & 16 & \texttt{0+\allowbreak{}1+\allowbreak{}2+\allowbreak{}3+\allowbreak{}5+\allowbreak{}0~1+\allowbreak{}0~2+\allowbreak{}2~5+\allowbreak{}0~1~2+\allowbreak{}0~1~3+\allowbreak{}0~1~4+\allowbreak{}0~2~3+\allowbreak{}0~2~4+\allowbreak{}0~3~4+\allowbreak{}0~3~5+\allowbreak{}1~2~3+\allowbreak{}1~2~4+\allowbreak{}1~2~5+\allowbreak{}1~3~4+\allowbreak{}1~4~5+\allowbreak{}2~4~5+\allowbreak{}3~4~5} & 11\\
129 & $\factory{63}{6}{3}$ & 16 & \texttt{0+\allowbreak{}2+\allowbreak{}3+\allowbreak{}4+\allowbreak{}0~1+\allowbreak{}1~4+\allowbreak{}2~3+\allowbreak{}2~4+\allowbreak{}0~1~2+\allowbreak{}0~1~3+\allowbreak{}0~4~5+\allowbreak{}1~4~5+\allowbreak{}2~4~5+\allowbreak{}3~4~5} & 10\\
130 & $\factory{63}{6}{3}$ & 12 & \texttt{0~1+\allowbreak{}2~3+\allowbreak{}4~5} & 9\\
131 & $\factory{63}{6}{3}$ & 12 & \texttt{0~1~2+\allowbreak{}3~4~5} & 13\\
132 & $\factory{63}{7}{3}$ & 17 & \texttt{0+\allowbreak{}1~2+\allowbreak{}1~3+\allowbreak{}1~4+\allowbreak{}1~5+\allowbreak{}1~6+\allowbreak{}1~2~3+\allowbreak{}1~2~4+\allowbreak{}1~2~5+\allowbreak{}1~3~6+\allowbreak{}1~4~6+\allowbreak{}2~3~4+\allowbreak{}3~4~5+\allowbreak{}3~4~6} & --\\
\bottomrule
\end{tabular}
\par\medskip
\begin{tabular}{@{}r l r p{0.50\textwidth} r@{}}
\multicolumn{5}{c}{TABLE~\thetable\ (continued)}\\[3pt]
\toprule
\# & factory & $N$ & output phase polynomial & $\Tgate$-count\\
\midrule
169 & $\factory{76}{2}{4}$ & 11 & \texttt{0+\allowbreak{}0~1} & 2\\
170 & $\factory{76}{4}{4}$ & 13 & \texttt{0+\allowbreak{}1+\allowbreak{}2+\allowbreak{}3+\allowbreak{}0~1+\allowbreak{}2~3} & 2\\
171 & $\factory{76}{4}{4}$ & 13 & \texttt{0~1~2+\allowbreak{}0~1~3} & 7\\
172 & $\factory{76}{5}{3}$ & 12 & \texttt{0+\allowbreak{}1+\allowbreak{}2+\allowbreak{}3~4} & 6\\
173 & $\factory{76}{5}{3}$ & 12 & \texttt{0+\allowbreak{}1+\allowbreak{}2~3~4} & 8\\
174 & $\factory{76}{5}{3}$ & 12 & \texttt{0+\allowbreak{}1~2+\allowbreak{}3~4} & 7\\
\bottomrule
\end{tabular}
\par\medskip
\begin{tabular}{@{}r l r p{0.50\textwidth} r@{}}
\multicolumn{5}{c}{TABLE~\thetable\ (continued)}\\[3pt]
\toprule
\# & factory & $N$ & output phase polynomial & $\Tgate$-count\\
\midrule
211 & $\factory{90}{11}{3}$ & 19 & \texttt{0+\allowbreak{}1+\allowbreak{}2+\allowbreak{}3+\allowbreak{}4+\allowbreak{}5+\allowbreak{}6+\allowbreak{}7+\allowbreak{}8~9~10} & --\\
212 & $\factory{91}{4}{3}$ & 11 & \texttt{0+\allowbreak{}1+\allowbreak{}2~3} & 5\\
213 & $\factory{91}{6}{3}$ & 13 & \texttt{0+\allowbreak{}1+\allowbreak{}2+\allowbreak{}3~4~5} & 9\\
214 & $\factory{91}{6}{3}$ & 13 & \texttt{0+\allowbreak{}1~2+\allowbreak{}3~4~5} & 10\\
215 & $\factory{91}{7}{3}$ & 14 & \texttt{0+\allowbreak{}1~2+\allowbreak{}3~4+\allowbreak{}5~6} & --\\
216 & $\factory{91}{7}{3}$ & 14 & \texttt{0+\allowbreak{}1~2~3+\allowbreak{}4~5~6} & --\\
\bottomrule
\end{tabular}
\par\medskip
\begin{tabular}{@{}r l r p{0.50\textwidth} r@{}}
\multicolumn{5}{c}{TABLE~\thetable\ (continued)}\\[3pt]
\toprule
\# & factory & $N$ & output phase polynomial & $\Tgate$-count\\
\midrule
253 & $\factory{127}{12}{3}$ & 19 & \texttt{0~1~2+\allowbreak{}3~4~5+\allowbreak{}6~7~8+\allowbreak{}9~10~11} & --\\
254 & $\factory{127}{16}{3}$ & 23 & \texttt{0+\allowbreak{}1~2~3+\allowbreak{}4~5~6+\allowbreak{}7~8~9+\allowbreak{}10~11~12+\allowbreak{}13~14~15} & --\\
255 & $\factory{127}{17}{3}$ & 24 & \texttt{0+\allowbreak{}1+\allowbreak{}2+\allowbreak{}3+\allowbreak{}4+\allowbreak{}5~6~7+\allowbreak{}8~9~10+\allowbreak{}11~12~13+\allowbreak{}14~15~16} & --\\
256 & $\factory{127}{26}{3}$ & 33 & \texttt{0+\allowbreak{}1+\allowbreak{}2+\allowbreak{}3+\allowbreak{}4+\allowbreak{}5+\allowbreak{}6+\allowbreak{}7+\allowbreak{}8+\allowbreak{}9+\allowbreak{}10+\allowbreak{}11~12~13+\allowbreak{}11~14~15+\allowbreak{}11~16~17+\allowbreak{}12~18~19+\allowbreak{}12~20~21+\allowbreak{}13~22~23+\allowbreak{}13~24~25+\allowbreak{}14~18~22+\allowbreak{}14~20~24+\allowbreak{}15~19~23+\allowbreak{}15~21~25+\allowbreak{}16~18~25+\allowbreak{}16~20~23+\allowbreak{}17~19~24+\allowbreak{}17~21~22} & --\\
257 & $\factory{127}{37}{3}$ & 44 & \emph{entangled output}, $1868$ factors & --\\
258 & $\factory{128}{3}{6}$ & 18 & \texttt{0~1~2} & 7\\
\bottomrule
\end{tabular}
\par\medskip
\endgroup\medskip

\end{document}